\documentclass[11pt]{article}
\usepackage{amsmath,amssymb}
\usepackage[dvips]{epsfig}
\usepackage{graphicx}

\newtheorem{lemma}{Lemma}

\newenvironment{proof}
{\begin{trivlist} \item[]{\bf Proof. }}%
{\hspace*{\fill}$\rule{.3\baselineskip}{.35\baselineskip}$\end{trivlist}}

\newcommand{\R}{\mathbb{R}}

\newcommand{\Z}{\mathbb{Z}}
\newcommand{\N}{\mathbb{N}}

\renewcommand{\geq}{\geqslant}
\renewcommand{\leq}{\leqslant}

\renewcommand{\phi}{\varphi}
\newcommand{\be}{\begin{eqnarray}}
\newcommand{\ee}{\end{eqnarray}}

\newcommand{\eps}{\varepsilon}

\begin{document}

\title{\bf Excited states in the Thomas--Fermi limit:\\ a variational approach}

\author{M. Coles$^1$, D.E. Pelinovsky$^1$, and P.G. Kevrekidis$^2$,  \\
{\small $^{1}$ Department of Mathematics and Statistics, McMaster
University, Hamilton, Ontario, Canada, L8S 4K1} \\
{\small $^{2}$ Department of Mathematics and Statistics, University
of Massachusetts, Amherst, MA 01003} }

\date{\today}
\maketitle

\begin{abstract}
Excited states of Bose--Einstein condensates are considered in the
semi-classical (Thomas-Fermi) limit of the Gross--Pitaevskii equation with repulsive
inter-atomic interactions and a harmonic potential.
The relative dynamics of dark solitons (density dips on the localized condensate)
with respect to the harmonic potential and to each other is approximated
using the averaged Lagrangian method. This permits a complete
characterization of the equilibrium positions of the 
dark solitons as a function of the chemical potential parameter.
It also yields an analytical handle on the oscillation frequencies 
of dark solitons around such equilibria. The asymptotic
predictions are generalized for an arbitrary number
of dark solitons and are corroborated by numerical
computations for 2- and 3-soliton configurations.
\end{abstract}

\section{Introduction}

The defocusing nonlinear Schr\"{o}dinger equation is a prototypical model
for a variety of different settings including nonlinear optics, liquids,
mechanical systems, and magnetic films, among others. In one spatial
dimension, its prototypical excitation is the dark soliton,
i.e., a localized density dip on a continuous-wave background
(carrying also a phase jump).

One of the major areas where the description of dark solitons
with a mean-field model (also known as the Gross-Pitaevskii equation) has been
the physics of atomic Bose-Einstein condensates (BECs) \cite{pethick,PitStr}.
There, the repulsive inter-atomic interactions can be accurately
captured by an effective nonlinear self-action \cite{revnonlin}.
A considerable volume of experimental work has conclusively
demonstrated the relevance of such nonlinear waveforms
within harmonically confined condensates.
Although in earlier works, such coherent structures were
dynamically or thermally unstable \cite{bpa,han2},
more recent work has overcome such limitations \cite{engels,technion,hambcol,our2}.
This has been achieved by working at
sufficiently low temperatures (of the order of $10$nK) and for
strongly confined in the transverse directions, cigar-shaped BECs.
Furthermore, in these recent experiments, the nature of the generation process
(e.g., by interference of two independent BECs \cite{technion,our2,our1},
or through interaction of the BEC with an appropriate light pulse
\cite{hambcol}), it has been possible to produce two or
more dark solitons on the background of a localized condensate.
In principle, the resulting number of dark solitons can be chosen at will,
as indicated in \cite{our2}.

These recent developments prompt us to examine the dynamics
of dark solitons which are harmonically confined within
localized repulsive Bose-Einstein condensates. These can be thought of as
density dips that arise in nonlinear variants of the excited
states of the quantum harmonic oscillator \cite{alfimov2}. The study of the equilibrium
positions and near-equilibrium dynamics of these density dips
is the principal theme of the present contribution. In particular,
using a Lagrangian (variational) approach, we compute the asymptotic dependence
on the chemical potential parameter both for equilibrium positions of dark solitons and for 
their oscillation frequencies around such equilibria.

This article is organized as follows. In Section 2, we present the general mathematical
setup of the problem. Section 3 examines the single soliton case,
Section 4 extends considerations to 2-solitons, and Section 5
generalizes the results to an arbitrary number of $m$-solitons for $m \geq 2$.
Section 6 compares our asymptotic predictions to numerical
computations and suggests some interesting directions for further
study.

\section{Mathematical Setup}

Let us start with the Gross--Pitaevskii equation with a harmonic potential and repulsive 
nonlinear interactions
\begin{equation}
\label{GPphys}
i v_{\tau}  = - \frac{1}{2} v_{\xi \xi} + \frac{1}{2} \xi^2 v + |v|^2 v  - \mu v,
\end{equation}
where $v(\xi,\tau) : \mathbb{R} \times \mathbb{R} \to \mathbb{C}$ is the wave function and
$\mu \in \R$ represents the chemical potential (and is physically associated with the number of atoms in
the condensate). We are interested in localized modes of the Gross--Pitaevskii equation
in the limit $\mu \to \infty$, which is associated with the semi-classical or
Thomas--Fermi limit. Using the scaling transformation,
\begin{equation}
v(\xi,t) = \mu^{1/2} u(x,t), \quad \xi = (2 \mu)^{1/2} x, \quad \tau = 2 t,
\label{translation}
\end{equation}
the Gross--Pitaevskii equation (\ref{GPphys}) is transformed to the semi-classical form
\begin{equation}
\label{GP} i \eps u_t + \eps^2 u_{xx} + (1 - x^2 - |u|^2) u = 0,
\end{equation}
where $u(x,t) : \mathbb{R} \times \mathbb{R} \to \mathbb{C}$ is a new wave function 
and $\eps = (2 \mu)^{-1}$ is a small parameter.

Let $\eta_{\eps}$ be a real positive solution of the stationary problem
\begin{equation}
\label{stationaryGP} \eps^2 \eta_\eps''(x) + (1- x^2 -
\eta_\eps^2(x)) \eta_{\eps}(x) =0, \quad x \in \mathbb{R}.
\end{equation}
Main results of Ignat \& Millot \cite{IM,IM2} and Gallo \&
Pelinovsky \cite{GalPel2} state that for any sufficiently small $\eps > 0$
there exists a smooth solution $\eta_{\eps} \in {\cal
C}^{\infty}(\mathbb{R})$ that decays to zero as $|x| \to \infty$
faster than any exponential function. The ground state
converges pointwise as $\eps \to 0$ to the compact Thomas--Fermi
cloud
\begin{equation}
\label{Thomas-Fermi}
\eta_0(x) := \lim_{\eps \to 0} \eta_{\eps}(x) = \left\{ \begin{array}{cl} (1 - x^2)^{1/2}, \;\; & \mbox{for} \;\; |x| < 1, \\
0, \;\; & \mbox{for} \;\; |x| > 1. \end{array} \right.
\end{equation}
Useful properties of the ground state $\eta_{\eps}$ for sufficiently small $\eps > 0$
are summarized as follows:

\begin{itemize}
\item For any compact subset $K \in (-1,1)$, there is $C_K > 0$ such that
\begin{equation}
\label{C-K-bound} \| \eta_{\eps} - \eta_0 \|_{C^1(K)} \leq C_K
\eps^2.
\end{equation}

\item There is $C > 0$ such that
\begin{equation}
\label{L-infty-bound} \| \eta_{\eps} - \eta_0 \|_{L^{\infty}} \leq
C \eps^{1/3}, \quad \| \eta_{\eps}' \|_{L^{\infty}} \leq C
\eps^{-1/3}, \quad \| \eta_{\eps}'' \|_{L^{\infty}} \leq C
\eps^{-1}.
\end{equation}
\end{itemize}

We shall consider excited states of the Gross--Pitaevskii equation
(\ref{GP}), which are non-positive solutions of the stationary
problem
\begin{equation}
\label{stationaryGPexc} \eps^2 u_\eps''(x) + (1- x^2 -
u_\eps^2(x)) u_{\eps}(x) =0, \quad x \in \mathbb{R}.
\end{equation}
The excited states can be classified by the number $m$ of zeros of $u_{\eps}(x)$ on $\R$.
A unique solution with $m$ zeros exists near $\eps = \eps_m$ by the local
bifurcation theory \cite{PelKev}, where $\eps_m = \frac{1}{1 + 2 m}$, $m \in
\mathbb{N}$. Because of the symmetry of the harmonic potential,
the $m$-th excited state $u_{\eps}(x)$ is even on $\mathbb{R}$ for even $m \in
\mathbb{N}$ and odd on $\mathbb{R}$ for odd $m \in \mathbb{N}$.
The $m$-th excited state is continued for $\eps < \eps_m$ numerically
by Zezyulin {\em et al.} \cite{ZAKP}.

In our work we shall apply variational approximations \cite{KK} to study
relative dynamics of dark solitons (localized
solutions with nonzero boundary conditions on the background of
the positive ground state $\eta_{\eps}$) with respect to the harmonic
potential and to each other. In particular, we obtain results on existence and
spectral stability of the excited states from analysis of 
equilibrium positions of dark solitons and their oscillation frequencies near 
such equilibrium.  To enable this formalism, we substitute
$$
u(x,t) = \eta_{\eps}(x) v(x,t)
$$
to the Gross--Pitaevskii equation (\ref{GP}) and find an
equivalent equation
\begin{equation}
\label{GP-renormalized} i \eps \eta_{\eps}^2 v_t + \eps^2 \left( \eta_{\eps}^2 v_x
\right)_x + \eta_{\eps}^4 (1 - |v|^2) v = 0.
\end{equation}
Excited states are solutions of the stationary equation
\begin{equation}
\label{stationaryGP-renormalized} \eps^2 \frac{d}{dx} \left(
\eta_{\eps}^2(x) V'_m(x) \right) + \eta_{\eps}^4(x) (1 - V^2_m(x))
V_m(x) = 0,  \quad x \in \mathbb{R},
\end{equation}
which have exactly $m$ zeros on $\R$ and satisfy the boundary
conditions
$$
\lim_{x \to \pm \infty} V_m(x) = (\pm 1)^m, \quad m \in
\mathbb{N}.
$$
Solutions of the stationary Gross--Pitaevskii equation
(\ref{stationaryGP-renormalized}) are critical points of the
energy functional
\begin{equation}
\label{energy-functional}
\Lambda(v) = \eps^2 \int_{\R} \eta_{\eps}^2(x) |v_x|^2 dx + \frac{1}{2}
\int_{\mathbb{R}} \eta_{\eps}^4(x) (1 - |v|^2)^2 dx.
\end{equation}
in the sense of $\frac{\delta \Lambda}{\delta \bar{v}} |_{v = V_m} = 0$.
The time-dependent Gross--Pitaevskii equation (\ref{GP-renormalized})
follows from the Lagrangian function $L(v) = K(v) + \Lambda(v)$,
where
\begin{equation}
\label{Lagrangian} K(v) = \frac{i}{2} \eps \int_{\mathbb{R}}
\eta_{\eps}^2(x) ( v \bar{v}_t - \bar{v} v_t ) dx,
\end{equation}
by means of the Euler--Lagrange equations
$$
\frac{\delta L}{\delta \bar{v}} - \frac{d}{dt} \frac{\delta
L}{\delta \bar{v}_t} = 0.
$$
In what follows, we obtain variational approximations for
time-dependent solutions near the excited states $V_m(x)$ for $m =
1$, $m = 2$, and in the general case $m \geq 2$. We also compare these approximations 
with numerical results for $m = 2$ and $m = 3$.

\section{1-soliton ($m = 1$)}
\label{section-1-soliton}

Let us consider the dark soliton
\begin{equation}
\label{1-soliton}
v_1(x,t) = A(t) \; \tanh\left( \eps^{-1} B(t) (x - a(t)) \right) + i b(t), \quad A > 0, \; B > 0, \; a \in \R, \; b \in \R,
\end{equation}
as an ansatz for the Lagrangian $L(v)$. The
motivation for this choice originates from the fact that
(\ref{1-soliton}) is an exact solution of (\ref{GP-renormalized})
if $\eta_{\eps} = 1$ under constraints
$$
A = \sqrt{1 - b^2}, \quad B = \frac{1}{\sqrt{2}} \sqrt{1-b^2}, \quad a = a_0 + \sqrt{2} b t,
\quad b = b_0,
$$
where $a_0 \in \R$ and $b_0 \in (-1,1)$ are arbitrary
$t$-independent parameters. In view of the relation
$$
|v_1|^2 = A^2 + b^2 - A^2 {\rm sech}^2\left( \eps^{-1} B(t) (x - a(t)) \right),
$$
it is clear that $a$ is a center of the dark soliton,
$b$ its speed, $A$ determines its amplitude, and $B$ determines its width.
If the dark soliton is placed inside the confinement of the compact Thomas--Fermi
cloud (\ref{Thomas-Fermi}), then the constraint $a \in (-1,1)$ has to be added.

When $\eta_{\eps} \neq 1$, the trial function (\ref{1-soliton}) is no longer an exact
solution of (\ref{GP-renormalized}) but it becomes the best
approximate solution if parameters $(A,B,a,b)$ are chosen from the
Euler--Lagrange equations of the averaged Lagrangian
$L_1(A,B,a,b) = L(v_1)$. This variational method provides
a useful qualitative approximation to physicists for understanding
the dynamics of dark solitons under perturbations \cite{KK}.
Unlike the work of \cite{KK}, we do not need to renormalize the
Lagrangian function $L(v)$ thanks to the rapidly decaying weight
function $\eta_{\eps}^2(x)$ under the integration sign in
(\ref{energy-functional})--(\ref{Lagrangian}).

Let us choose $A = \sqrt{1 - b^2}$ to satisfy the boundary conditions
$$
\lim_{x \to \pm \infty} |v_1(x,t)| = 1 \quad \mbox{\rm for all}
\quad t \in \mathbb{R}.
$$
Substitution of ansatz (\ref{1-soliton}) to $L(v)$ and
integration in $\R$ results in the effective Lagrangian
\begin{eqnarray}
\nonumber
L(v_1) & = & \frac{\eps \dot{b}}{\sqrt{1-b^2}} \int_{\R}
\eta_{\eps}^2(x) \tanh(z) dx
+ b \sqrt{1-b^2} B \dot{a} \int_{\R} \eta_{\eps}^2(x) {\rm sech}^2(z) dx \\
\nonumber
& \phantom{t} & - \eps b \sqrt{1-b^2} \dot{B} B^{-1} \int_{\R} \eta_{\eps}^2(x) z {\rm sech}^2(z) dx
+ (1 - b^2) B^2 \int_{\R} \eta_{\eps}^2(x) {\rm sech}^4(z) dx \\
& \phantom{t} &  + \frac{1}{2} (1 - b^2)^2 \int_{\R}
\eta_{\eps}^4(x) {\rm sech}^4(z) dx, \label{averaged-Lagrangian-1}
\end{eqnarray}
where $z = \eps^{-1} B(x-a)$. Note the pointwise limits
\begin{equation}
\label{pointwise-bound}
\lim_{\eps \to 0} \tanh(z) = {\rm sign}(x-a), \quad
\lim_{\eps \to 0} {\rm sech}^2(z) = 0, \quad x \in \R \backslash\{0\},
\end{equation}
which show that $\lim_{\eps \to 0} L(v_1) = 0$. The
value of $L(v_1)$ in the limit of $\eps \to 0$ is computed in the following lemma.

\begin{lemma}
Assume that $B > 0$ and $a \in (-1,1)$. Then,
\begin{eqnarray*}
L_1 := \lim_{\eps \to 0} \frac{L(v_1)}{2\eps} & = &
-\frac{\dot{b}}{\sqrt{1-b^2}} (a - \frac{1}{3} a^3) + b \sqrt{1 -
b^2} (1 - a^2) \dot{a} \\
& \phantom{t} & + \frac{2}{3} (1 - a^2) (1 - b^2) B + \frac{1}{3B}
(1 - a^2)^2 (1 - b^2)^2.
\end{eqnarray*}
\label{proposition-1-soliton}
\end{lemma}

\begin{proof}
Thanks to the limit (\ref{Thomas-Fermi}), the pointwise bound (\ref{pointwise-bound}),
and the Dominated Convergence Theorem, we have
$$
\lim_{\eps \to 0} \int_{\R} \eta_{\eps}^2(x) \tanh(z) dx =
\int_{-1}^1 (1 - x^2) {\rm sign}(x-a) dx = -2a + \frac{2}{3} a^3,
$$
To compute the remaining four integrals in (\ref{averaged-Lagrangian-1}), we use the change of variables $x \to z$, so that
\begin{eqnarray*}
\eps^{-1} B \int_{\R} \eta_{\eps}^2(x) {\rm sech}^2(z) dx & = &
\int_{\R} \eta_{\eps}^2\left(a + \eps z B^{-1}\right) {\rm
sech}^2(z) dz \\ & = & \int_{z_-}^{z_+} \eta_0^2\left(a + \eps z
B^{-1}\right) {\rm sech}^2(z) dz + \eps^{1/3} \int_{\R}
R_{\eps,B,a}(z) {\rm sech}^2(z) dz,
\end{eqnarray*}
where $z_{\pm} = \eps^{-1} B (\pm 1 - a)$ and the reminder term
satisfies the uniform bound $\| R_{\eps,B,a} \|_{L^{\infty}} \leq
C$ for some $C > 0$, thanks to the first bound (\ref{L-infty-bound}). As a result, the second term does
not contribute to the limit $\eps \to 0$. To deal with the first
term, we decompose the integral into three parts
\begin{eqnarray*}
(1-a^2) \int_{z_-}^{z_+} {\rm sech}^2(z) dz - 2 \eps a B^{-1}
\int_{z_-}^{z_+} z {\rm sech}^2(z) dz - \eps^2 B^{-2}
\int_{z_-}^{z_+} z^2 {\rm sech}^2(z) dz.
\end{eqnarray*}
We recall that the integral
$$
\int_{\alpha \eps^{-1}}^{\infty} z^k {\rm sech}^2(z) dz, \quad k \geq 0,
$$
is exponentially small in $\eps$ if $\alpha > 0$ is $\eps$-independent. As
a result, the second and third terms do not contribute to the limit $\eps \to 0$,
while the first term gives
$$
\lim_{\eps \to 0} \eps^{-1} B \int_{\R} \eta_{\eps}^2(x) {\rm
sech}^2(z) dx = (1-a^2) \int_{\R} {\rm sech}^2(z) dz =
2(1-a^2).
$$
The remaining three integrals in (\ref{averaged-Lagrangian-1})
are computed similarly to the second integral in (\ref{averaged-Lagrangian-1}) and give
\begin{eqnarray*}
\lim_{\eps \to 0} \eps^{-1} B \int_{\R} \eta_{\eps}^2(x) z {\rm sech}^2(z) dx
& = & (1-a^2) \int_{\R} z {\rm sech}^2(z) dz = 0, \\
\lim_{\eps \to 0} \eps^{-1} B \int_{\R} \eta_{\eps}^2(x) {\rm
sech}^4(z) dx & = & (1-a^2) \int_{\R} {\rm sech}^4(z) dz = \frac{4}{3}(1-a^2), \\
\lim_{\eps \to 0} \eps^{-1} B \int_{\R} \eta_{\eps}^4(x) {\rm
sech}^4(z) dx & = & (1-a^2)^2 \int_{\R} {\rm sech}^4(z) dz = \frac{4}{3}(1-a^2)^2.
\end{eqnarray*}
Combining all individual computations gives the result for $L_1$.
\end{proof}

Since $\dot{B}$ is absent in $L_1 := L_1(a,b,B)$, variation of $L_1$
with respect to $B$ gives an algebraic equation on $B$ with the exact solution
$$
B = \frac{1}{\sqrt{2}} \sqrt{1 - a^2} \sqrt{1 - b^2}.
$$
Eliminating $B$ from $L_1(a,b,B)$, we simplify the effective Lagrangian to the form
$$
L_1(a,b) = \frac{2 \sqrt{2}}{3} (1-a^2)^{3/2} (1-b^2)^{3/2} - 2
\sqrt{1-b^2} \dot{b} (a - \frac{1}{3} a^3) + \frac{d}{dt} \left[(a
- \frac{1}{3} a^3) b \sqrt{1-b^2} \right],
$$
where the last term is the full derivative. Since adding a full
derivative does not change the Euler--Lagrange equations, the last
term can be dropped from $L_1$. Variation with respect to $a$ and
$b$ give the following system of equations
$$
\dot{a} = \sqrt{2} \sqrt{1-a^2} b, \quad \dot{b} = -\frac{\sqrt{2} a (1-b^2)}{\sqrt{1-a^2}},
$$
which is equivalent to the linear oscillator equation
$$
\ddot{a} + 2 a = 0.
$$
The critical point $(a,b) = (0,0)$ corresponds to the
solution $V_1$ of the stationary equation (\ref{stationaryGP-renormalized}). Oscillations near the critical
point with frequency $\sqrt{2}$ corresponds to the oscillations of
the dark soliton $V_1$ relative to the positive ground state $\eta_{\eps}$ in the Thomas--Fermi
limit $\eps \to 0$; see e.g. \cite{KP04} and references therein.
This frequency was found to be the smallest nonzero
frequency in the spectrum of the spectral stability problem 
associated with the first excited state, see
Fig. 2 in \cite{PelKev}.

\section{2-solitons ($m = 2$)}
\label{section-2-soliton}

Let us now consider a superposition of two dark
solitons
\begin{eqnarray}
\nonumber v_2(x,t) & = & \left[ A_1(t) \; \tanh\left( \eps^{-1}
B_1(t) (x - a_1(t)) \right) + i b_1(t) \right] \\
& \phantom{t} & \times \left[ A_2(t) \; \tanh\left( \eps^{-1}
B_2(t) (x - a_2(t)) \right) + i b_2(t) \right],  \label{2-soliton}
\end{eqnarray}
where we shall use the relations for the individual dark solitons
$$
A_j = \sqrt{1 - b_j^2}, \quad B_j = \frac{1}{\sqrt{2}} \sqrt{1 - a_j^2} \sqrt{1
- b_j^2}, \quad j = 1,2.
$$
In-phase oscillations of two dark solitons are very similar to the oscillations
of one dark soliton and have the same frequency, as we will show in Section \ref{section-3-soliton}.
Therefore, we shall consider out-of-phase oscillations of two dark solitons
and choose
$$
a_1 = -a, \quad a_2 = a, \quad b_1 = -b,
\quad b_2 = b,
$$
with $a \in (0,1)$ and $b \in \R$. Substitution of $v_2$ to $\Lambda(v)$ gives
\begin{eqnarray*}
\Lambda(v_2) & = & A^2 B^2 \int_{\R} \eta_{\eps}^2(x) \left[ {\rm sech}^4 (z_+) +
{\rm sech}^4 (z_-) - 2 b^2 {\rm sech}^2 (z_+) {\rm sech}^2 (z_-) \right. \\
& \phantom{t} & \left. - A^2 {\rm sech}^2 (z_+) {\rm sech}^2 (z_-) \left( {\rm sech}^2 (z_+)
+ {\rm sech}^2 (z_-) - 2 \tanh (z_+) \tanh (z_-) \right) \right] dx \\
& + & \frac{1}{2} A^4 \int_{\R} \eta_{\eps}^4(x) \left[ {\rm sech}^4 (z_+) + {\rm sech}^4 (z_-) + 2 {\rm sech}^2 (z_+) {\rm sech}^2 (z_-) \right. \\
& \phantom{t} & \left. - 2 A^2 {\rm sech}^2 (z_+) {\rm sech}^2 (z_-) \left( {\rm sech}^2 (z_+)
+ {\rm sech}^2 (z_-) \right) + A^4 {\rm sech}^4 (z_+) {\rm sech}^4 (z_-) \right] dx,
\end{eqnarray*}
where $z_{\pm} = \eps^{-1} B (x \pm a)$. The integrals that only depend
on $z_+$ or $z_-$ are computed similarly to the case of $1$-soliton. The overlapping integrals that depend on both
$z_+$ and $z_-$ are computed under the apriori assumption
\begin{equation}
\label{assumption-derivation} a \leq C_1 \eps^{1/6}, \quad e^{-4 B a \eps^{-1}} \leq C_2 \eps^2 \log(\eps),
\end{equation}
for some $C_1,C_2 > 0$ and sufficiently small $\eps > 0$. As we will see later, the apriori
assumption allows us to recover the equilibrium state of two dark solitons
and to study perturbations near the equilibrium.

After simplifications, one can write
\begin{eqnarray*}
\Lambda_2 := \frac{\Lambda(v_2)}{2\eps} = \Lambda_+ + \Lambda_- +
\Lambda_{\rm overlap},
\end{eqnarray*}
where
\begin{eqnarray*}
\Lambda_{\pm} := \frac{A^2 B^2}{2 \eps} \int_{\R} \eta_{\eps}^2(x)
{\rm sech}^4(z_{\pm}) dx + \frac{A^4}{4 \eps} \int_{\R}
\eta_{\eps}^4(x) {\rm sech}^4(z_{\pm}) dx
\end{eqnarray*}
and
\begin{eqnarray*}
\Lambda_{\rm overlap} & = & -\frac{A^2 B^2}{2\eps} \int_{\R} \eta_{\eps}^2(x)
{\rm sech}^2 (z_+) {\rm sech}^2 (z_-)
\\ & \phantom{t} & \phantom{text} \times \left[ 2 b^2 + A^2
\left( {\rm sech}^2 (z_+) + {\rm sech}^2 (z_-) - 2 \tanh (z_+) \tanh (z_-) \right) \right] dx \\
& \phantom{t} & + \frac{A^4}{4\eps} \int_{\R} \eta_{\eps}^4(x) {\rm
sech}^2 (z_+) {\rm sech}^2 (z_-) \\
& \phantom{t} & \phantom{text} \times \left[ 2 - 2 A^2 \left( {\rm sech}^2
(z_+) + {\rm sech}^2 (z_-) \right) + A^4 {\rm sech}^2 (z_+) {\rm
sech}^2 (z_-) \right] dx.
\end{eqnarray*}
The terms $\Lambda_{\pm}$ are the potential energies of the
individual dark solitons and the term $\Lambda_{\rm overlap}$
contains overlapping integrals. By Lemma
\ref{proposition-1-soliton}, we have
\begin{eqnarray*}
\Lambda_{\pm} = \frac{4 (1 - a^2)^{3/2} (1 -
b^2)^{3/2}}{3 \sqrt{2}}  + {\cal O}(\eps^{1/3}).
\end{eqnarray*}
The overlapping integrals for small
$\eps$ are computed in the following lemma.

\begin{lemma}
\label{proposition-2-soliton}
Assume that $a \in (0,1)$ satisfies (\ref{assumption-derivation}), $b \in \R$, and
$$
A = \sqrt{1 - b^2}, \quad B = \frac{1}{\sqrt{2}} \sqrt{1 - a^2}
\sqrt{1 - b^2}.
$$
Then,
\begin{eqnarray*}
\Lambda_{\rm overlap} = - 8 \sqrt{2} (1 - a^2)^{3/2} (1 -
b^2)^{5/2} \; e^{-4 B a \eps^{-1}} \left( 1 + {\cal
O}(\eps^{1/3}) \right).
\end{eqnarray*}
\end{lemma}

\begin{proof}
To compute the overlapping integrals, we use the symmetry of the integrand and the change of
variables $x \to z_-$. The first overlapping integral in $\Lambda_{\rm overlap}$ is given by
\begin{eqnarray*}
\eps^{-1} B \int_{\R} \eta^2_{\eps}(x) {\rm sech}^2 (z_+) {\rm sech}^2 (z_-) dx =
2 \int_{-Ba \eps^{-1}}^{\infty}
\eta_{\eps}^2 \left( a + \eps z B^{-1}\right) {\rm sech}^2 (z) {\rm sech}^2(z + 2 B a \eps^{-1}) dz,
\end{eqnarray*}
where $z \equiv z_-$. Similarly to the proof of Lemma \ref{proposition-1-soliton}, we break the integral
into four parts
\begin{eqnarray*}
2 (1 - a^2) \int_{-Ba \eps^{-1}}^{B (1-a) \eps^{-1}}
{\rm sech}^2 (z) {\rm sech}^2(z + 2 B a \eps^{-1}) dz \\
- 4 a \eps B^{-1} \int_{-B a \eps^{-1}}^{B (1-a) \eps^{-1}}
z {\rm sech}^2 (z) {\rm sech}^2(z + 2 B a \eps^{-1}) dz \\
- 2 \eps^2 B^{-2} \int_{-B a \eps^{-1}}^{B (1-a) \eps^{-1}}
z^2 {\rm sech}^2 (z) {\rm sech}^2(z + 2 B a \eps^{-1}) dz \\
+ 2 \eps^{1/3} \int_{-B a \eps^{-1}}^{\infty}
R_{\eps,B,a}(z) {\rm sech}^2 (z) {\rm sech}^2(z + 2 B a \eps^{-1}) dz,
\end{eqnarray*}
where the reminder term satisfies the bound $\| R_{\eps,B,a}
\|_{L^{\infty}} \leq C$ for some $C > 0$, thanks to the bound (\ref{L-infty-bound}).
The first part gives the leading
order of the integral according to the explicit calculation
\begin{eqnarray*}
I_1 & = & \int_{-B a \eps^{-1}}^{B (1-a) \eps^{-1}}
{\rm sech}^2 (z) {\rm sech}^2(z + 2 B a \eps^{-1}) dz \\
& = & 16 \left( \int_{-B a \eps^{-1}}^{B a \eps^{-1}} + \int_{B a \eps^{-1}}^{B (1-a) \eps^{-1}} 
\right) \frac{e^{-4z -4B a \eps^{-1}}}{(1+e^{-2z})^2 (1 + e^{-2z-4B a \eps^{-1}})^2} dz
\end{eqnarray*}
We have
$$
0 \leq e^{-2z-4B a \eps^{-1}} \leq e^{-2B a \eps^{-1}}, \quad z \geq -B a \eps^{-1},
$$
and
$$
e^{-B (1-a) \eps^{-1}} \ll e^{-B a \eps^{-1}}, \quad a \leq C \eps^{1/6},
$$
so that
\begin{eqnarray*}
I_1 & = & 16 e^{-4B a \eps^{-1}} \left( \int_{-B a \eps^{-1}}^{B a \eps^{-1}}
\frac{e^{-4z}}{(1+e^{-2z})^2} dz \right) \left( 1 + {\cal O}\left(e^{-2B a \eps^{-1}}\right) \right) 
+ {\cal O}\left(e^{-8 B a \eps^{-1}}\right) \\
& = & 8 e^{-4B a \eps^{-1}}
\left( 2 B a \eps^{-1} - 1 \right) \left( 1 + {\cal O}\left(e^{-2B a \eps^{-1}}\right) \right).
\end{eqnarray*}
The second part of the overlapping integral is computed from the explicit computation
\begin{eqnarray*}
I_2 & = & a \eps \int_{-B a \eps^{-1}}^{B (1-a) \eps^{-1}} z {\rm sech}^2 (z) {\rm sech}^2(z + 2 B a \eps^{-1}) dz \\
& = & a \eps \left( \int_{-B a \eps^{-1}}^{B a \eps^{-1}} + \int_{B a \eps^{-1}}^{B (1-a) \eps^{-1}} \right)
z {\rm sech}^2 (z) {\rm sech}^2(z + 2 B a \eps^{-1}) dz \\
& = & {\cal O}(a^2 I_1) + {\cal O}\left(e^{-6B a \eps^{-1}}\right) = {\cal O}(a^2 I_1), \\
\end{eqnarray*}
The last two parts of the overlapping integrals are computed similarly and yield
\begin{eqnarray*}
I_3 & = & \eps^2 \int_{-B a \eps^{-1}}^{B (1-a) \eps^{-1}}
z^2 {\rm sech}^2 (z) {\rm sech}^2(z + 2 B a \eps^{-1}) dz = {\cal O}(a^2 I_1), \\
I_4 & = & \eps^{1/3} \int_{-B a \eps^{-1}}^{\infty}
R_{\eps,B,a}(z) {\rm sech}^2 (z) {\rm sech}^2(z + 2 B a \eps^{-1}) dz = {\cal O}(\eps^{1/3} I_1).
\end{eqnarray*}
Under the assumption (\ref{assumption-derivation}), we have
$$
e^{-2B a \eps^{-1}} = {\cal O}(\eps \log^{1/2}(\eps)) \quad \mbox{\rm and} \quad
a^2 = {\cal O}(\eps^{1/3}),
$$
so that we finally obtain
$$
\eps^{-1} B \int_{\R} \eta^2_{\eps}(x) {\rm sech}^2 (z_+) {\rm sech}^2 (z_-) dx =
16 (1 - a^2) e^{-4B a \eps^{-1}} \left( 2 B a \eps^{-1} - 1 \right)
\left( 1 + {\cal O}(\eps^{1/3}) \right).
$$
Similarly, we compute the other overlapping integrals in $\Lambda_{\rm overlap}$ as follows:
\begin{eqnarray*}
& \phantom{t} & \eps^{-1} B \int_{\R} \eta_{\eps}^2(x) {\rm
sech}^2 (z_+) {\rm sech}^2 (z_-) \left( {\rm sech}^2 (z_+) + {\rm
sech}^2 (z_-) \right) dx    \\
& \phantom{t} & \phantom{texttext} = \frac{64}{3} (1 - a^2)
e^{-4B a \eps^{-1}} \left( 1 + {\cal O}(\eps^{1/3}) \right),
\end{eqnarray*}
\begin{eqnarray*}
& \phantom{t} & \eps^{-1} B \int_{\R} \eta_{\eps}^2(x) {\rm
sech}^2 (z_+) {\rm sech}^2 (z_-) {\rm tanh} (z_+) {\rm tanh} (z_-)
dx \\ & \phantom{t} & \phantom{texttext} = 32 (1 - a^2)
e^{-4B a \eps^{-1}} \left( - B a \eps^{-1} + 1 \right)
\left( 1 + {\cal O}(\eps^{1/3}) \right),
\end{eqnarray*}
and
\begin{eqnarray*}
\eps^{-1} B \int_{\R} \eta_{\eps}^4(x) {\rm sech}^4 (z_+) {\rm sech}^4 (z_-) dx
= 512 (1 - a^2)^2 e^{-8 B a \eps^{-1}} \left( B a \eps^{-1} - \frac{11}{12} \right)
\left( 1 + {\cal O}(\eps^{1/3}) \right).
\end{eqnarray*}
Combining these computations together, we obtain the expression
for $\Lambda_{\rm overlap}$.
\end{proof}

Variations of $\Lambda_2(a,b)$ define critical points that
correspond to the solution $V_2(x)$ of the stationary
equation (\ref{stationaryGP-renormalized}). Since $\Lambda_2$
is even in $b \in \R$, the set of critical points includes $b =
0$. Note that $v_2(x,t)$ in (\ref{2-soliton}) is real if $b = 0$,
which agree with $V_2(x)$ being real-valued.

Since $\Lambda_+ + \Lambda_-$ is even in $a$ and the
overlapping integral is small under assumption
(\ref{assumption-derivation}), variation of $\Lambda_2(a,0)$
in $a$ gives a root finding problem
\begin{equation}
\label{nonlinear-equation-a}
-4 \sqrt{2} \eps a \left( 1 + {\cal O}(\eps^{1/3}) \right) +
32 e^{-2\sqrt{2} a \eps^{-1}} \left( 1 + {\cal O}(\eps^{1/3}) \right) = 0.
\end{equation}
The asymptotic
analysis of the roots of the nonlinear equation
(\ref{nonlinear-equation-a}) in the following lemma shows that the apriori
assumption (\ref{assumption-derivation}) is indeed satisfied.

\begin{lemma}
For sufficiently small $\eps > 0$, there exists a simple root
of the nonlinear equation  (\ref{nonlinear-equation-a})
in the neighborhood of $0$, which is expanded by
\begin{equation}
\label{asymptotic-a-result}
a = \frac{\eps}{\sqrt{2}} \left( -\log(\eps) - \frac{1}{2}
\log|\log(\eps)| + \frac{3}{2} \log(2) + o(1) \right) \quad {\rm as} \quad \eps \to 0.
\end{equation}
\label{proposition-roots}
\end{lemma}

\begin{proof}
Taking a natural logarithm of the nonlinear equation (\ref{nonlinear-equation-a}), we obtain
$$
2 \sqrt{2} a + \eps \log(a) = -\eps \log(\eps) + \frac{5}{2} \eps
\log(2) + {\cal O}(\eps^{4/3}).
$$
Let $a = -\frac{1}{\sqrt{2}} \eps \log(\eps) U$ and rewrite the
problem for $U$:
\begin{equation}
\label{U-leading-order}
U - \frac{\log(U)}{2 \log(\eps)} =  1 + \frac{\log|\log(\eps)|}{2
\log(\eps)} - \frac{3\log(2)}{2 \log(\eps)} \left( 1 + {\cal
O}(\eps^{1/3}) \right).
\end{equation}
By the Implicit Function Theorem applied to
equation (\ref{U-leading-order}), existence of a unique root $U(\eps)$ in a
one-sided neighborhood of $\eps > 0$ is proved, where $U(\eps)$ is
continuous in $\eps > 0$ and $\lim_{\eps \downarrow 0} U(\eps) =
1$. To estimate the remainder term for $|U(\eps) -1|$, one can further decompose
$$
U = 1 + \frac{\log|\log(\eps)|}{2 \log(\eps)} (1 + V)
$$
and rewrite the problem for $V$:
\begin{equation}
\label{V-leading-order}
V - \frac{\log\left( 1 + \frac{\log|\log(\eps)|}{2 \log(\eps)} (1
+ V) \right)}{\log|\log(\eps)|} = - \frac{3 \log(2)}{\log
|\log(\eps)|} \left( 1 + {\cal O}(\eps^{1/3}) \right).
\end{equation}
By the Implicit Function Theorem applied again to equation (\ref{V-leading-order}),
existence of a unique root $V(\eps)$ in a
one-sided neighborhood of $\eps > 0$ is proved, where $V(\eps)$ is
continuous in $\eps > 0$ and $\lim_{\eps \downarrow 0} V(\eps) =
0$. Substitution of $U$ back to formula for $a$ gives (\ref{asymptotic-a-result}).
\end{proof}

By Lemma \ref{proposition-roots}, we can study temporal dynamics of two dark solitons near the bound
state that corresponds to a small root of the nonlinear equation
(\ref{nonlinear-equation-a}).

To proceed with time-derivative terms, we substitute
(\ref{2-soliton}) to the kinetic part $K(v)$ in (\ref{Lagrangian})
and find that
$$
K_2 := \frac{K(v_2)}{2\eps} = K_+ + K_- + K_{\rm overlap},
$$
where
\begin{eqnarray*}
K_{\pm} & = & \mp \frac{\dot{b}}{2 \sqrt{1-b^2}}
\int_{\R} \eta_{\eps}^2(x) \tanh(z_{\pm}) dx + \frac{b
\sqrt{1-b^2} B \dot{a}}{2\eps}
\int_{\R} \eta_{\eps}^2(x) {\rm sech}^2(z_{\pm}) dx \\
& \phantom{t} & \pm \frac{b \sqrt{1-b^2} \dot{B}}{2 B}
\int_{\R} \eta_{\eps}^2(x) z_{\pm} {\rm sech}^2(z_{\pm}) dx
\end{eqnarray*}
and
\begin{eqnarray*}
K_{\rm overlap} & = & \frac{1}{2} \dot{b} (1-b^2)^{1/2}
\int_{\R}
\eta_{\eps}^2(x) \left( \tanh(z_+) {\rm sech}^2(z_-) - \tanh(z_-) {\rm sech}^2(z_+) \right) dx \\
& \phantom{t} & - \eps^{-1} b (1-b^2)^{3/2} (B
\dot{a} + \dot{B} a) \int_{\R} \eta_{\eps}^2(x) {\rm
sech}^2(z_+) {\rm sech}^2(z_-) dx.
\end{eqnarray*}
The terms $K_{\pm}$ are the kinetic energies of the individual
dark solitons and the term $K_{\rm overlap}$ contains overlapping
integrals. By Lemma \ref{proposition-1-soliton}, we have
$$
\lim_{\eps \to 0} (K_+ + K_-) = - 4 \sqrt{1-b^2} \dot{b}
(a - \frac{1}{3} a^3) + 2 \frac{d}{dt} \left[(a -
\frac{1}{3} a^3) b \sqrt{1-b^2} \right].
$$
The overlapping integrals for small
$\eps$ are estimated in the following lemma.

\begin{lemma}
\label{proposition-3-soliton}
Assume that $a \in (0,1)$ satisfies (\ref{assumption-derivation}), $b \in \R$, and
$$
A = \sqrt{1 - b^2}, \quad B = \frac{\sqrt{1 - a^2}
\sqrt{1 - b^2}}{\sqrt{2}}.
$$
Then,
\begin{eqnarray*}
K_{\rm overlap} & = & 2 \eps \dot{b} (1-b^2)^{1/2} B^{-1} (1 - a^2)
\left( 1 + {\cal O}(\eps^{1/3}) \right) \\
& \phantom{t} & - 16 b (1-b^2)^{3/2} (\dot{a} + B^{-1} \dot{B} a) (1 - a^2)
 e^{-4B a \eps^{-1}} \left( 2 B a \eps^{-1} - 1 \right)
\left( 1 + {\cal O}(\eps^{1/3}) \right).
\end{eqnarray*}
\end{lemma}

\begin{proof}
The first and second terms in $K_{\rm
overlap}$ are estimated similarly to the proof of Lemma \ref{proposition-2-soliton}.
Note that the first term disappears in the limit $\eps \to 0$.
\end{proof}

To obtain effective dynamical equations on $(a,b)$
valid in the domain specified by assumption (\ref{assumption-derivation}),
we expand $L_2(a,b) = \frac{L(v_2)}{2\eps}$ in the quadratic form
in $(a,b)$ and apply the limit $\eps \to 0$ to all but
the overlapping integrals. As a result,
the reduced effective Lagrangian $L_2(a,b)$ takes the form
\begin{eqnarray*}
L_2(a,b) & \sim & \frac{4 \sqrt{2}}{3} \left( 1 -
\frac{3}{2} (b^2 + a^2) + {\cal O}(b^2 + a^2)^2 \right)
- 4 a \dot{b} \left( 1 + {\cal O}(b^2 + a^2) \right) \\
& \phantom{t} & - 8 \sqrt{2} e^{-2 \sqrt{2} a \eps^{-1} (1 + {\cal O}(b^2 + a^2))}
\left( 1 + {\cal O}(b^2 + a^2) \right).
\end{eqnarray*}
In variables $(a,b)$, the Euler--Lagrange equations
at the leading order become
$$
\dot{a} = \sqrt{2} b, \quad \dot{b} = - \sqrt{2} a +
8 \eps^{-1} e^{-2\sqrt{2} a \eps^{-1}},
$$
or, equivalently, recover the nonlinear oscillator equation
$$
\ddot{a} + 2 a = 8 \sqrt{2} \eps^{-1}
e^{-\frac{2\sqrt{2} a}{\eps}}.
$$
The equilibrium state is given by the root $a_0(\eps)$ of the nonlinear
equation (\ref{nonlinear-equation-a}). This equilibrium state is a center and
linear oscillations near the center satisfy
$$
\ddot{\delta} + \omega_0^2 \delta = 0,
$$
where $\delta = a - a_0(\eps)$ and
\begin{eqnarray}
\nonumber
\omega_0^2(\eps) & = & 2 + \frac{32}{\eps^2} e^{-2\sqrt{2}
a_0(\eps) \eps^{-1}} = 2 + \frac{4 \sqrt{2} a_0(\eps)}{\eps} \\
\label{frequency} & = & -4 \log(\eps)
- 2 \log|\log(\eps)| + 2 + 6 \log(2) + o(1), \quad \mbox{\rm as} \quad \eps \to 0,
\end{eqnarray}
thanks to Lemma \ref{proposition-roots}. We note that the frequency $\omega_0(\eps)$ of out-of-phase oscillations
of two dark solitons grows in the limit $\eps \to 0$. This property will be further discussed in Section \ref{section-numerics}.

\section{$m$-solitons with $m \geq 2$}
\label{section-3-soliton}

We extrapolate the results of the previous section to the case of $m$-solitons with $m \geq 2$. The general
superposition of $m$ dark solitons is substituted in the form
\begin{eqnarray}
v_m(x,t) = \prod_{j = 1}^m \left( A_j(t) \; \tanh\left( \eps^{-1}
B_j(t) (x - a_j(t)) \right) + i b_j(t) \right),  \label{m-soliton}
\end{eqnarray}
where
$$
A_j = \sqrt{1 - b_j^2}, \quad B_j = \frac{1}{\sqrt{2}} \sqrt{1 - a_j^2} \sqrt{1
- b_j^2}, \quad j \in \{ 1,2,...,m\}.
$$
Under the same assumptions of
$$
|a_j| \leq C \eps^{1/6}, \quad j \in \{ 1,2,...,m\}
$$
and
$$
e^{-\sqrt{2} (a_{j+1}-a_j) \eps^{-1}} \leq C \eps^2 \log(\eps), \quad
j \in \{ 1,2,...,m-1 \},
$$
for some $C > 0$, we reduce the effective Lagrangian $L_m := \frac{L(v_m)}{2 \eps}$ to
the leading order
$$
L_m \sim -\sqrt{2} \sum_{j=1}^m \left( a_j^2 + b_j^2 \right) - 2 \sum_{j=1}^m a_j \dot{b}_j
- 8 \sqrt{2} \sum_{j=1}^{m-1} e^{-\sqrt{2} (a_{j+1}-a_j) \eps^{-1}},
$$
where only the quadratic terms in $(a_j,b_j)$ and only
the pairwise interaction potentials are taken into account. Using the Euler--Lagrange equations,
we obtain
\begin{equation}
\label{system-a-b-j}
\dot{a}_j = \sqrt{2} b_j, \quad \dot{b}_j = -\sqrt{2} a_j - 8 \eps^{-1} \left(
e^{-\sqrt{2} (a_{j+1}-a_j) \eps^{-1}} - e^{-\sqrt{2} (a_j-a_{j-1}) \eps^{-1}} \right), \quad
j \in \{1,2,...,m\},
\end{equation}
where boundary conditions $a_0 = -\infty$ and $a_{m+1} = \infty$ must be used. The center of mass
$\langle a \rangle = \frac{1}{m} \sum_{j = 1}^m a_j$ satisfies the linear oscillator equation
\begin{equation}
\label{lin-oscil}
\ddot{\langle a \rangle} + 2 \langle a \rangle = 0,
\end{equation}
which recovers the frequency of oscillations of a $1$ dark soliton in Section
\ref{section-1-soliton}. Let us introduce the set of normal coordinates
$$
x_j = \sqrt{2} (a_{j+1}-a_j) \eps^{-1}, \quad j \in \{1,2,...,m-1\},
$$
and rewrite system (\ref{system-a-b-j}) in the scalar form
\begin{equation}
\label{Toda-lattice}
\ddot{x}_j + 2 x_j + 16 \eps^{-2} \left( e^{-x_{j+1}} - 2 e^{-x_j} + e^{-x_{j-1}} \right) = 0, \quad
j \in \{1,2,...,m-1\},
\end{equation}
where the boundary conditions are now $x_0 = x_m = \infty$. System (\ref{Toda-lattice})
is known as the Toda lattice with nonzero masses, which is not integrable by inverse
scattering (unlike its counterpart with zero masses). We are only interested
in existence of critical points in the Toda lattice and in the
distribution of eigenvalues in the linearization around the critical points.

Critical points of the
Toda lattice (\ref{Toda-lattice}) are defined by solutions
of system of algebraic equations
\begin{equation}
\label{critical-points-Toda}
2 x_j + 16 \eps^{-2} \left( e^{-x_{j+1}} - 2 e^{-x_j} + e^{-x_{j-1}} \right) = 0, \quad
j \in \{1,2,...,m-1\}.
\end{equation}
Let the $(m-1)\times(m-1)$ matrix ${\bf A}$ be given by
$$
{\bf A} = \left[ \begin{array}{ccccccc} 2 & -1 & 0 & 0 & \cdots & 0 & 0\\
-1 & 2 & -1 & 0 & 0  & \cdots & 0 \\ 0 & -1 & 2 & -1 & \cdots & 0 & 0 \\
\vdots & \vdots & \vdots & \vdots  & \vdots & \vdots & \vdots \\
0 & 0 & 0 & 0  & ... & -1 & 2 \end{array} \right].
$$
Matrix ${\bf A}$ arises in the central-difference approximation of the second
derivative subject to the Dirichlet boundary conditions. It is strictly positive and
thus invertible. The system of algebraic equations (\ref{critical-points-Toda}) can
be written in the matrix-vector form
\begin{equation}
\label{critical-point-vector-form}
{\bf A} e^{-{\bf x}} = \frac{\eps^2}{8} {\bf x} \quad \Rightarrow \quad
e^{-{\bf x}} = \frac{\eps^2}{8} {\bf A}^{-1} {\bf x}.
\end{equation}
Solutions of system (\ref{critical-point-vector-form}) in the limit $\eps \to 0$
are analyzed in the following lemma.

\begin{lemma}
\label{lemma-Toda}
For sufficiently small $\eps > 0$, there exists a unique solution of system (\ref{critical-point-vector-form})
in the neighborhood of $\infty$, which is expanded by
\begin{equation}
\label{critical-points-leading}
{\bf x} = - 2 \log(\eps) {\bf 1} - \log|\log(\eps)| {\bf 1} + 2 \log(2) {\bf 1} - \log({\bf A}^{-1} {\bf 1}) +
o(1), \quad \mbox{\rm as} \quad \eps \to 0,
\end{equation}
where ${\bf 1} = [1,1,...,1]^T \in \R^{m-1}$.
\end{lemma}

\begin{proof}
Applying the natural logarithm to system (\ref{critical-point-vector-form}), we rewrite the system as follows
$$
{\bf x} = - 2 \log(\eps) {\bf 1} + 3 \log(2) {\bf 1} - \log\left( {\bf A}^{-1} {\bf x} \right).
$$
Repeating the proof of Lemma \ref{proposition-roots}, we find the desired expansion (\ref{critical-points-leading}).
\end{proof}

Back to the physical variables $(a_1,...,a_m)$, the result of Lemma \ref{lemma-Toda} implies
that the coordinates of dark solitons are centered $\langle a \rangle = 0$ and distributed 
with nearly equal spacing as $\eps \to 0$. Linearizing the Toda lattice (\ref{Toda-lattice}) about the root of system
(\ref{critical-points-Toda}), we obtain the linear eigenvalue problem
\begin{equation}
\label{eigenvalue-general}
(2 - \omega^2) \xi_j - 16 \eps^{-2} \left( e^{-x_{j+1}} \xi_{j+1} - 2 e^{-x_j} \xi_j
+ e^{-x_{j-1}} \xi_{j-1} \right) = 0, \quad j \in \{1,2,...,m-1\},
\end{equation}
where $\xi_0$ and $\xi_m$ are not determined because the coefficients
in front of $\xi_0$ and $\xi_m$ are zero. Using the representation (\ref{critical-point-vector-form}),
we rewrite the linear eigenvalue problem in the form
\begin{equation}
\label{eigenvalue-general-explicit}
(2 - \omega^2) \xi_j - 2 \left( ({\bf A}^{-1} {\bf x})_{j+1} \xi_{j+1} - 2
({\bf A}^{-1} {\bf x})_{j} \xi_j + ({\bf A}^{-1} {\bf x})_{j-1} \xi_{j-1} \right)
= 0, \quad j \in \{1,2,...,m-1\}.
\end{equation}
Frequencies of oscillations are analyzed in the limit $\eps \to 0$ in the following lemma.

\begin{lemma}
\label{lemma-Toda-eigenvalues}
For sufficiently small $\eps > 0$, $(m-1)$ eigenvalues of the linear problem 
(\ref{eigenvalue-general-explicit}) are expanded by
\begin{equation}
\label{distribution-eigenvalues}
\omega^2 = 2 + \left( - 4 \log(\eps) - 2 \log|\log(\eps)| + 4 \log(2) \right) \Omega^2 + {\cal O}(1),
\end{equation}
where $\Omega^2 \in \left\{ 1, 3, 6, ..., \frac{m(m-1)}{2} \right\}$ and $m \geq 2$.
\end{lemma}

\begin{proof}
Let $\Omega^2$ be eigenvalues of the reduced eigenvalue problem
\begin{equation}
\label{eigenvalue-reduced}
\Omega^2 \xi_j + v_{j+1} \xi_{j+1} - 2 v_j \xi_j
+ v_{j-1} \xi_{j-1} = 0, \quad j \in \{1,2,...,m-1\},
\end{equation}
where ${\bf v} = {\bf A}^{-1} {\bf 1} \in \R^{m-1}$. We will show that all eigenvalues
of the reduced eigenvalue problem (\ref{eigenvalue-reduced}) are simple and given explicitly by
$\Omega^2 \in \left\{ 1, 3, 6, ..., \frac{m(m-1)}{2} \right\}$. If this is the case, 
the asymptotic expansion (\ref{critical-points-leading}) and the regular perturbation
theory for the matrix eigenvalue problem (\ref{eigenvalue-general-explicit}) imply that
$$
\left| \omega^2 - 2 + \left( 4 \log(\eps) + 2 \log|\log(\eps)| - 4 \log(2) \right) \Omega^2 \right|
= {\cal O}(1), \quad \mbox{\rm as} \quad \eps \to 0,
$$
for each eigenvalue $\Omega^2$.

To obtain the exact distribution of eigenvalues of 
the reduced eigenvalue problem (\ref{eigenvalue-reduced}),
we will find the vector ${\bf v}$ explicitly. The components of ${\bf v}$ satisfy the Dirichlet problem for second-order difference equations
$$
2 v_j - v_{j+1} - v_{j-1} = 1, \quad j \in \{1,2,...,m-1\},
$$
subject to $v_0 = v_m = 0$. The exact solution of this problem is
$$
v_j = \frac{1}{2} j (m-j), \quad j \in \{1,2,...,m-1\}.
$$
Let $k = j - \frac{m}{2}$, so that
$k \in {\cal I}_m := \{-\frac{m}{2}+1,-\frac{m}{2} + 2, ..., \frac{m}{2}-1\}$.
Note that ${\cal I}_m$ includes integer values for even $m$ and half-integer values for odd $m$.
Denote $\zeta_k = \xi_j$, and $\lambda = 2 \Omega^2$ and rewrite the reduced eigenvalue problem
(\ref{eigenvalue-reduced}) in the following explicit form
\begin{equation}
\label{eigenvalue-explicit}
\lambda \zeta_k = \left( \frac{m^2}{4} - k^2 \right) \left(
2 \zeta_k - \zeta_{k+1} - \zeta_{k-1} \right) + 2k \left( \zeta_{k+1} - \zeta_{k-1} \right)
+ \left( \zeta_{k+1} + \zeta_{k-1} \right), \quad k \in {\cal I}_m.
\end{equation}
First, we consider the problem (\ref{eigenvalue-explicit}) for all $k \in \mathbb{Z}$ 
with a fixed $m \geq 2$ and prove that there exists a basis of eigenvectors
$\mbox{\boldmath $\zeta$} \in \{ {\bf P}_n \}_{n \in \N_0}$ in the space of analytic functions on $\mathbb{Z}$
for an infinite set of eigenvalues $\lambda \in \{ (n+1)(n+2) \}_{n \in \mathbb{N}_0}$, where 
$\mathbb{N}_0 := \{0,1,2,...\}$. The corresponding
eigenvector $\mbox{\boldmath $\zeta$} = {\bf P}_n$ 
for each eigenvalue $\lambda = (n+1)(n+2)$ is given by the polynomial $P_n(k)$ in the form
\begin{equation}
\label{representation-polynomial}
\zeta_k = P_n(k) := k^n + c_1 k^{n-1} + c_2 k^{n-2} + ... + c_n, \quad k \in \mathbb{Z},
\end{equation}
with uniquely determined coefficients $(c_1,c_2,...,c_n)$. To show this, we note that
if $\mbox{\boldmath $\zeta$} \in {\cal P}_n$, where ${\cal P}_n$ is the vector space
of polynomials of degree $n$, then the vector field of the eigenvalue problem (\ref{eigenvalue-explicit})
belongs to ${\cal P}_n$. This follows from the fact that if $\mbox{\boldmath $\zeta$} \in {\cal P}_n$, then
\begin{equation}
\label{parities}
\left( 2 \zeta_k - \zeta_{k+1} - \zeta_{k-1} \right) \in {\cal P}_{n-2}, \quad
\left( \zeta_{k+1} - \zeta_{k-1} \right) \in {\cal P}_{n-1}, \quad
\left( \zeta_{k+1} + \zeta_{k-1} \right) \in {\cal P}_{n}.
\end{equation}
Substituting the representation (\ref{representation-polynomial}) to the linear eigenvalue problem
(\ref{eigenvalue-explicit}), we collect coefficients in front of $k^n$ to find that $\lambda = (n+1)(n+2)$
and the coefficients in front of $k^{n-1}$, $k^{n-2}$, ..., $k^0$ to find a lower triangular system
of linear equations for $c_1$, $c_2$, ..., $c_n$. The lower triangular coefficient matrix is invertible
(non-singular) because, if this is not the case, a homogeneous solution would exist 
to give a polynomial of a lower degree for the same eigenvalue $\lambda$. This contradicts to the fact that
the set $\{ (n+1)(n+2) \}_{n \in \mathbb{N}_0}$ includes only simple eigenvalues. Therefore,
a unique value for $(c_1,c_2,...,c_n)$ exists for a given $n$. All eigenvectors are linearly independent
since polynomials of different degrees defined on $\mathbb{Z}$ are linearly independent. The set of
all eigenvectors gives a basis of eigenvectors in the space of analytic functions on $\mathbb{Z}$.

Finally, we will prove that the basis of eigenvectors for the linear eigenvalue problem (\ref{eigenvalue-explicit})
on ${\cal I}_m$ with $m \geq 2$ is given by $\{ {\bf P}_0, {\bf P}_1,...,{\bf P}_{m-2} \}$,
which corresponds to the first $(m-1)$ eigenvalues $\lambda \in \{ 2,6,...,m(m-1)\}$. This follows from the fact
that each polynomial ${\bf P}_j$ is nonzero on ${\cal I}_m$ for $j \in \{0,1,...,m-2\}$ in the sense of
\begin{equation}
\label{condition-polynomial}
\sum_{k \in {\cal I}_m} |P_j(k)| \neq 0, \quad j \in \{0,1,...,m-2\}.
\end{equation}
By a contradiction, assume that condition (\ref{condition-polynomial}) is false, that is $P_j(k)$ has
$(m-1)$ roots on $\mathbb{R}$. However, $j < (m-1)$ and by the Fundamental Theorem
of Algebra, $P_j(k) \equiv 0$ for all $k \in \mathbb{Z}$, which is a contradiction. Therefore, condition
(\ref{condition-polynomial}) is satisfied. Furthermore, since polynomials $\{ {\bf P}_0, {\bf P}_1,...,{\bf P}_{m-2} \}$ correspond to distinct eigenvalues, these eigenvectors are linearly independent and form a basis
of eigenvectors on ${\cal I}_m$. This imply that all other polynomials in the set $\{ {\bf P}_j \}_{j \geq m-1}$
are linearly dependent from $\{ {\bf P}_0, {\bf P}_1,...,{\bf P}_{m-2} \}$ on ${\cal I}_m$, which means,
in view of different degrees and distinct eigenvalues, that ${\bf P}_j$ are identically zero on ${\cal I}_m$ for
all $j \geq m-1$. Therefore, the basis of eigenvectors for the linear eigenvalue problem (\ref{eigenvalue-explicit})
on ${\cal I}_m$ with $m \geq 2$ is given by $\{ {\bf P}_0, {\bf P}_1,...,{\bf P}_{m-2} \}$.
\end{proof}

We note that the polynomials $P_j(k)$ in the proof of Lemma \ref{lemma-Toda-eigenvalues}
are even in $k \in \Z$ for even $j$ and odd in $k \in \Z$ for odd $j$. This follows from
the parity transformations of operators in (\ref{parities}) and the explicit form of the 
linear eigenvalue problem (\ref{eigenvalue-explicit}). For example, let $m = 4$ so that ${\cal I}_4 = \{-1,0,1\}$ and
compute eigenvectors and eigenvalues of (\ref{eigenvalue-explicit}) explicitly:
\begin{eqnarray*}
& \phantom{t} & \lambda = 1 : \quad \zeta_k = P_0(k) = 1, \\
& \phantom{t} & \lambda = 3 : \quad \zeta_k = P_1(k) = k, \\
& \phantom{t} & \lambda = 6 : \quad \zeta_k = P_2(k) = k^2 - \frac{3}{5}.
\end{eqnarray*}
For the same case $m = 4$, $P_3(k) = k(k^2-1)$ so that $P_3(k) = 0$ for all $k \in {\cal I}_4$.

We finish this section with the explicit asymptotic approximations for $3$-solitons ($m = 3$). 
By the symmetry of system (\ref{critical-points-Toda}) with $m = 3$,
we understand that
$$
x_1 = x_2 = \sqrt{2} a \eps^{-1} \quad \Leftrightarrow \quad a_1 = -a, \;\; a_2 = 0, \;\; a_3 = a,
$$
where $a$ is a root of equation
$$
a - 4 \sqrt{2} \eps^{-1} e^{-\sqrt{2} a \eps^{-1}} = 0,
$$
which is expanded asymptotically as
\begin{equation}
\label{expansion-a-3-soliton}
a = \frac{\eps}{\sqrt{2}} \left( - 2 \log(\eps) - \log|\log(\eps)| + 2 \log(2) + o(1) \right), \quad
\mbox{\rm as} \quad \eps \to 0.
\end{equation}
Comparison with the asymptotic expansion (\ref{critical-points-leading}) shows that
$\log({\bf A}^{-1} {\bf 1}) = {\bf 0}$ or ${\bf v} = {\bf 1}$, which means that
the asymptotic distribution of frequencies (\ref{distribution-eigenvalues})
becomes accurate for $m = 3$ with ${\cal O}(1)$ replaced by $o(1)$. As a result, we find
asymptotic expansions of the two frequencies of out-of-phase oscillations near the
$3$-soliton equilibrium state in the form:
\begin{eqnarray}
\label{asympt-fre-1}
\left\{ \begin{array}{l}
\omega^2 = 2 + \left( - 4 \log(\eps) - 2 \log|\log(\eps)| + 4 \log(2) \right) + o(1), \\
\omega^2 = 2 + 3 \left( - 4 \log(\eps) - 2 \log|\log(\eps)| + 4 \log(2) \right) + o(1).
\end{array} \right.
\end{eqnarray}
These asymptotic results will be tested numerically in Section \ref{section-numerics}.

\section{Numerical results}
\label{section-numerics}

We now compare the asymptotic results with direct numerical
results for the existence and spectral stability of 2- and 3-soliton configurations.
We identify the relevant branches of stationary solutions by solving
the ordinary differential equation
\begin{equation}
\label{ODEphys}
- \frac{1}{2} v''(\xi) + \frac{1}{2} \xi^2 v(\xi) + v^3(\xi)  - \mu v(\xi) = 0, \quad \xi \in \R.
\end{equation}
A fixed point method (Newton-Raphson iteration) is used to solve a
discretized boundary-value problem, after a centered-difference scheme is
applied to the second-order derivatives with a typical spacing of
$\Delta \xi = 0.025$. The resulting solutions $v(\xi)$ are obtained
starting from the corresponding linear eigenfunction (with 2-
or 3-nodes at the linear limit) and continuation over the values of the
chemical potential parameter $\mu$ is used in order to
extend the branch to the large values of $\mu$. Note that
the existence and spectral stability of the $1$-soliton configuration
were examined in our earlier work in \cite{PelKev}.

Once the stationary solution is obtained for each
value of $\mu$, we linearize around it, using an ansatz of the form:
\begin{eqnarray}
v(\xi,\tau) = v(\xi) + \delta \left( a(\xi) e^{\lambda \tau} +
\bar{b}(\xi) e^{\bar{\lambda} \tau} \right),
\label{linearize}
\end{eqnarray}
where $\delta$ denotes a formal (small) parameter. The admissible values of
$\lambda$ (eigenvalues) are found from the condition that $(a,b) \in L^2(\R)$
is a solution of the linear eigenvalue problem
\begin{eqnarray}
\label{lin-spectral}
\left\{ \begin{array}{l}
- \frac{1}{2} a''(\xi) + \frac{1}{2} \xi^2 a(\xi) - \mu a(\xi) + v^2(\xi) (2 a(\xi) + b(\xi)) =
i \lambda a(\xi), \\
- \frac{1}{2} b''(\xi) + \frac{1}{2} \xi^2 b(\xi) - \mu b(\xi) + v^2(\xi) (a(\xi) + 2 b(\xi)) =
-i \lambda b(\xi). \end{array} \right.
\end{eqnarray}
Using again a discretization of the differential operators on the same
grid, we reduce (\ref{lin-spectral})
to a matrix eigenvalue problem which can be solved through standard numerical
linear algebra routines.

Our main results are summarized in Figures \ref{dfig1}-\ref{dfig2}
for the 2-soliton configuration and Figures \ref{dfig3}-\ref{dfig4}
for the 3-soliton case.

Fig. \ref{dfig1} compares the numerical result (solid line)
for the location of zeros of $v(\xi)$ to 
the asymptotic expansion (\ref{asymptotic-a-result}) (dash-dotted line), 
where the scaling transformation (\ref{translation}) has been taken into
account to translate the results from $\eps$ to $\mu$ by $\eps =
(2 \mu)^{-1}$. One can see that the asymptotic expansion yields a highly accurate
approximation of the numerical result. This is also
evidenced by the right panel of the figure comparing
the numerical solution $v(\xi)$ for $\mu=17$ (solid line) with
the variational ansatz (dashed line).

Fig. \ref{dfig2} shows the smallest eigenvalues of the linear eigenvalue
problem (\ref{lin-spectral}) obtained numerically (solid line). The resulting
eigenvalues can be classified into two types. The first one consists of a
countable
set of pairs of purely imaginary eigenvalues that give frequencies of
oscillations
of the ground state. The main result in Gallo \& Pelinovsky \cite{GalPel1}
states that
the frequencies of oscillations of the ground state $\eta_{\eps}$
are found in the limit $\eps \to 0$ as follows
$$
\lim_{\eps \to 0} \omega_n(\eps) = \sqrt{2 n (n+1)}, \quad n \geq 1.
$$
Note that $\omega_1(\eps) = 2$ is preserved for any $\eps > 0$ thanks to the
symmetry of the Gross--Pitaevskii equation with a harmonic
potential \cite{PelKev}.
Using the scaling transformation (\ref{translation}), we conclude that
these frequencies
satisfy the asymptotic limit
\begin{equation}
\label{asym-limits}
\lim_{\mu \to \infty} {\rm Im}(\lambda) = \frac{\sqrt{n (n+1)}}{\sqrt{2}}, \quad n \geq 1.
\end{equation}
The asymptotic limits (\ref{asym-limits}) are
shown on Fig. \ref{dfig2} by dashed lines.

The second set of eigenvalues consists of only two pairs of eigenvalues
and is associated with the
relative motions of the dark solitons \cite{our2}.
One pair of eigenvalues corresponds
to in-phase oscillations with frequencies ${\rm Im}(\lambda) \sim \frac{1}{\sqrt{2}}$
as $\mu \to \infty$ (or $\omega \sim \sqrt{2}$ as $\eps \to 0$
in notations of the linear oscillator equation (\ref{lin-oscil})).
The other pair of
eigenvalues corresponds to out-of-phase oscillations and it is
characterized by the asymptotic
expansion (\ref{frequency}). The asymptotic predictions for the second
set of frequencies are
shown by the dash-dotted lines.

The right panel of Fig. \ref{dfig2} shows the real part of the eigenvalues
close to the limit of local bifurcation at $\mu = \frac{5}{2}$. The instability, which was
studied in \cite{ZAKP}, is caused by the resonance between the out-of-phase $2$-soliton oscillations
and the quadrupolar oscillation mode of the ground state. Contrary to what is claimed in
numerical work of \cite{ZAKP}, we can see from Fig. \ref{dfig2} that the instability interval is finite
and the $2$-soliton excited state may be linearly stable for sufficiently
large values of the chemical potential $\mu$.

We note, however, that the frequency $\omega_0(\eps)$ of the out-of-phase oscillations of two
dark solitons given by the asymptotic expansion (\ref{frequency}) grows as $\eps \to 0$.
As a result, this frequency will coalesce with other frequencies $\omega_n(\eps)$, $n \geq 3$
associated with oscillations of the ground state as $\eps \to 0$.
Coalescence with the frequency $\omega_3(\eps)$ does not produce
an instability, because of the different parity of the corresponding eigenfunctions.
However, coalescence with the frequency $\omega_4(\eps)$ will produce the instability
again and it will happen roughly at $\eps \sim e^{-10}$. This value
of $\eps$ is too small to be confirmed by our numerical results on Fig. \ref{dfig2}.
This secondary instability of the $2$-soliton excited state is anticipated in a tiny interval
near $\eps \sim e^{-10}$, after which the neutrally stable frequency $\omega_0(\eps)$ will reappear
until further such coalescence occurrences arise
with frequencies $\omega_6(\eps)$, $\omega_8(\eps)$, etc.

\begin{figure}
\begin{center}
\includegraphics[width=8cm,height=6cm]{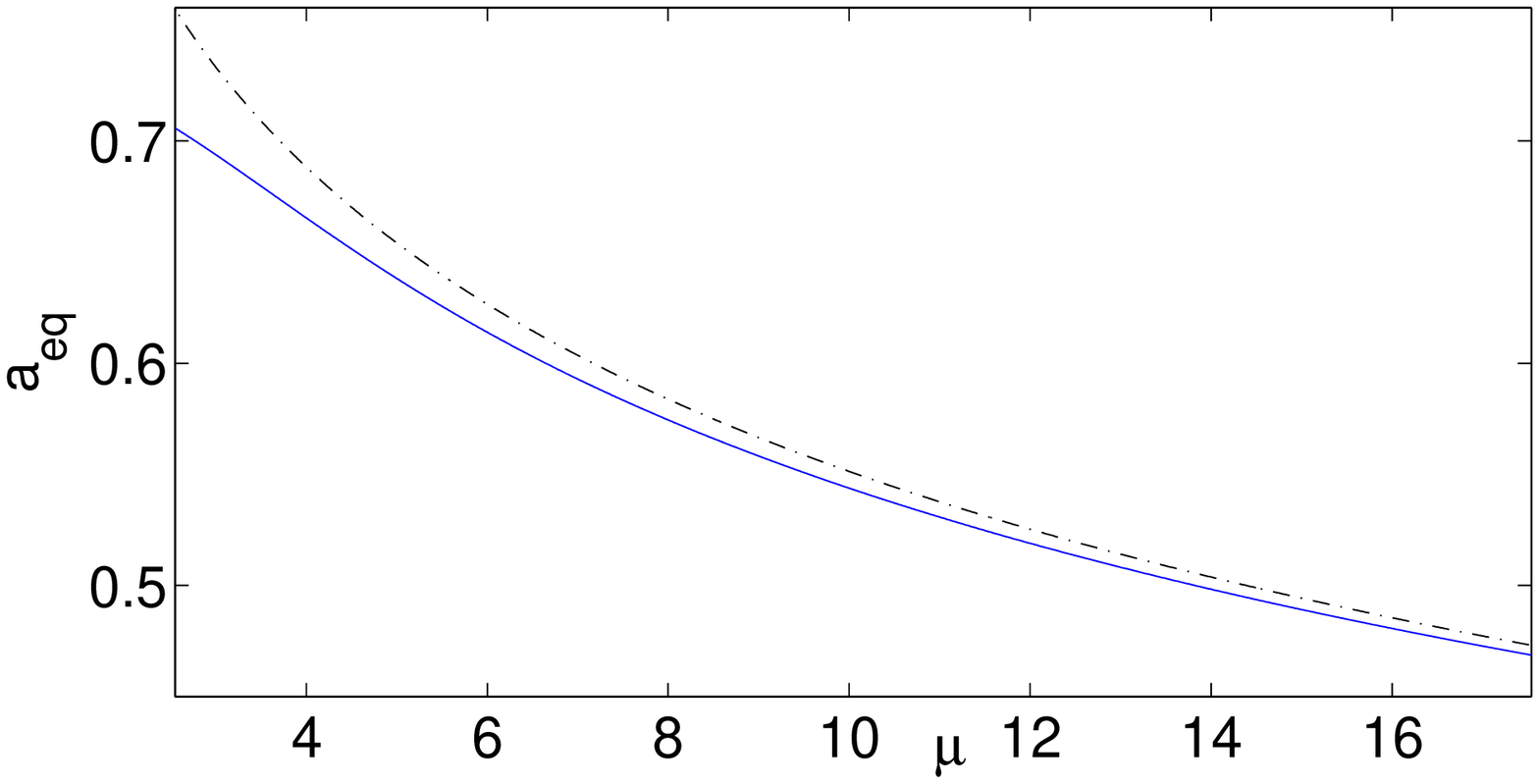}
\includegraphics[width=8cm,height=6cm]{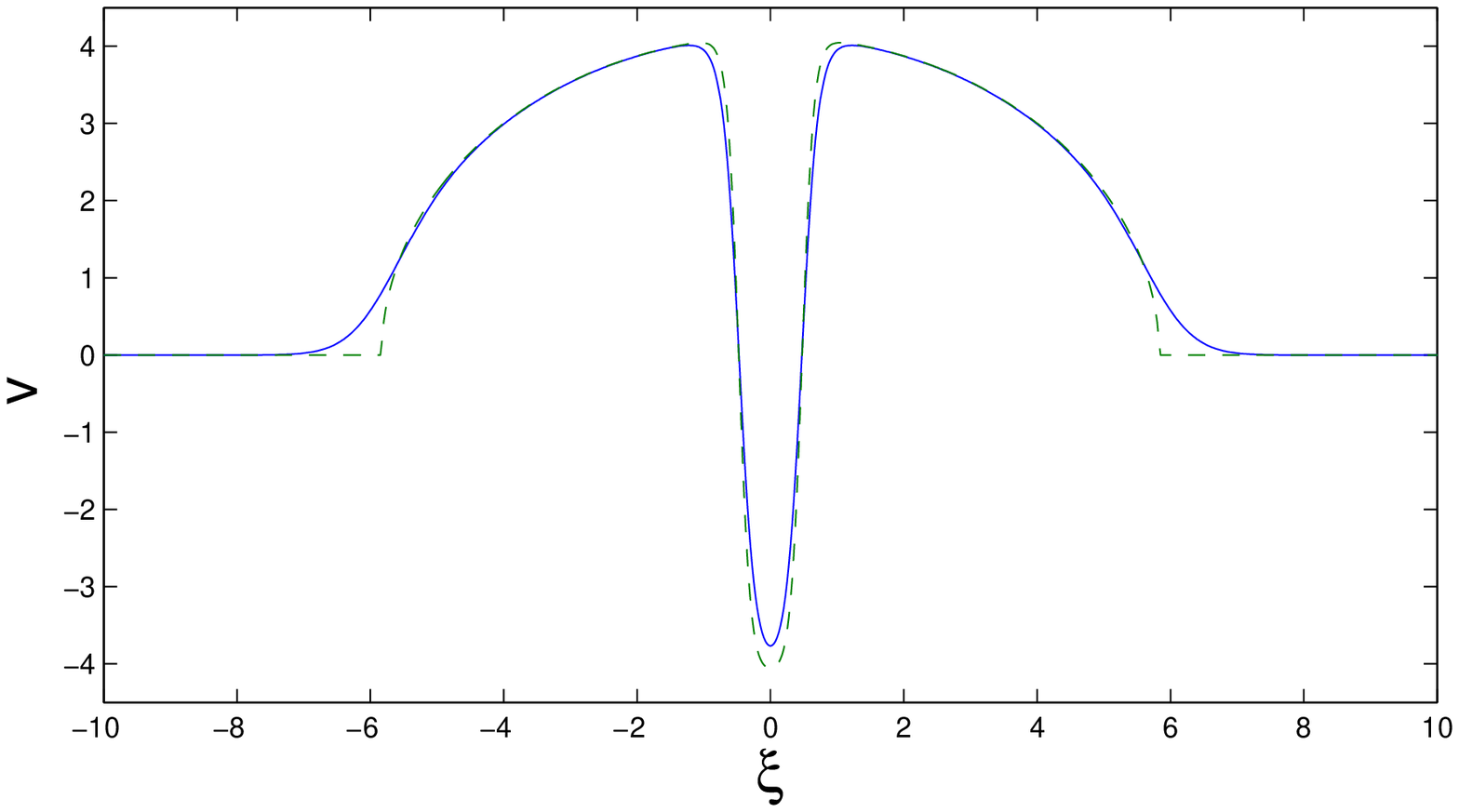}
\end{center}
\caption{Left: the equilibrium position of the two dark solitons versus
the chemical potential $\mu$. The solid line shows the direct
numerical result and the dash-dotted line represents the asymptotic approximation
(\ref{asymptotic-a-result}). Right: the solid
line shows the numerical solution $v(\xi)$ for $\mu=17$, while the dashed line
represents
the corresponding variational ansatz.} \label{dfig1}
\end{figure}
\begin{figure}
\begin{center}
\includegraphics[width=8cm,height=6cm]{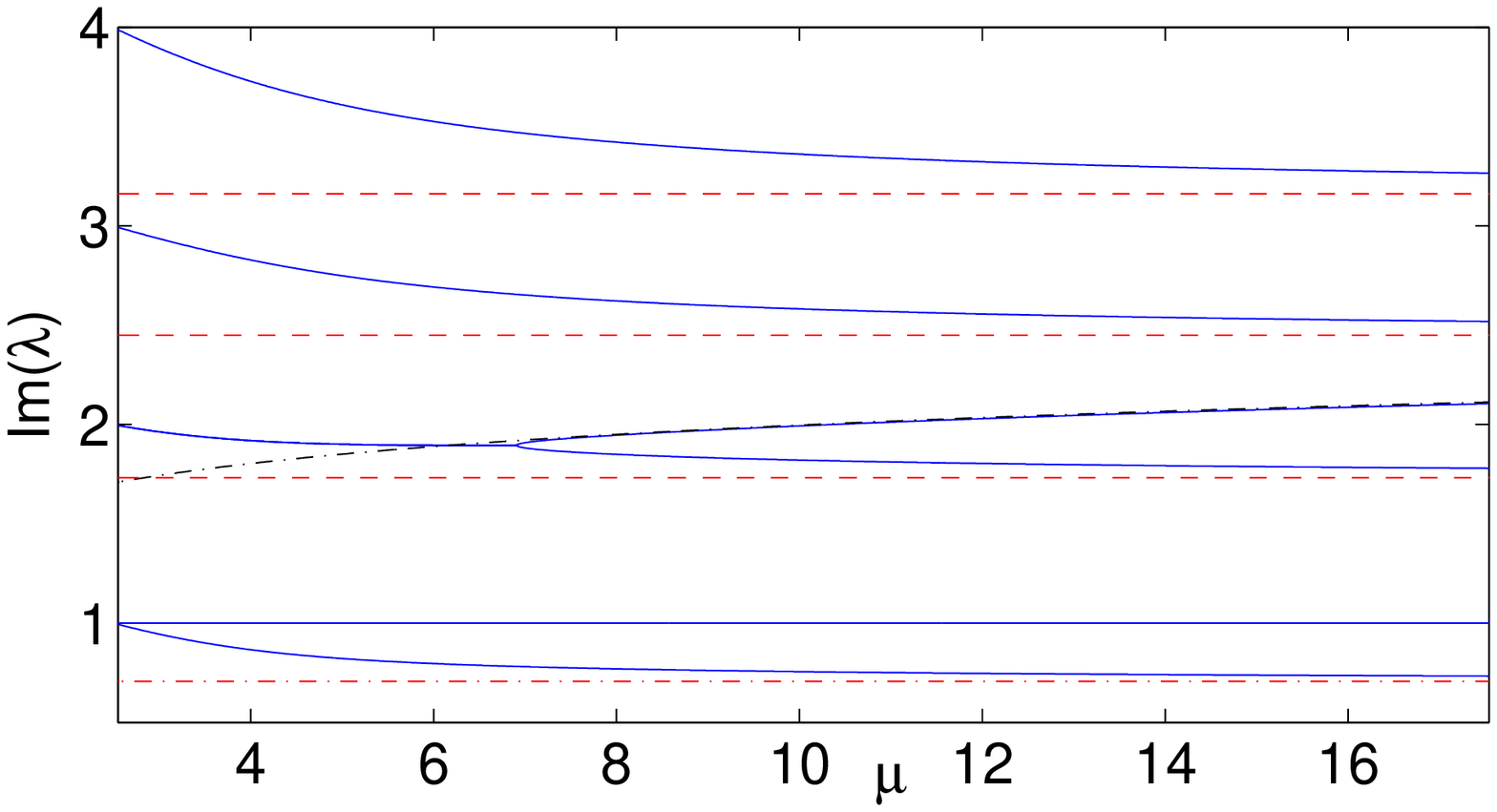}
\includegraphics[width=8cm,height=6cm]{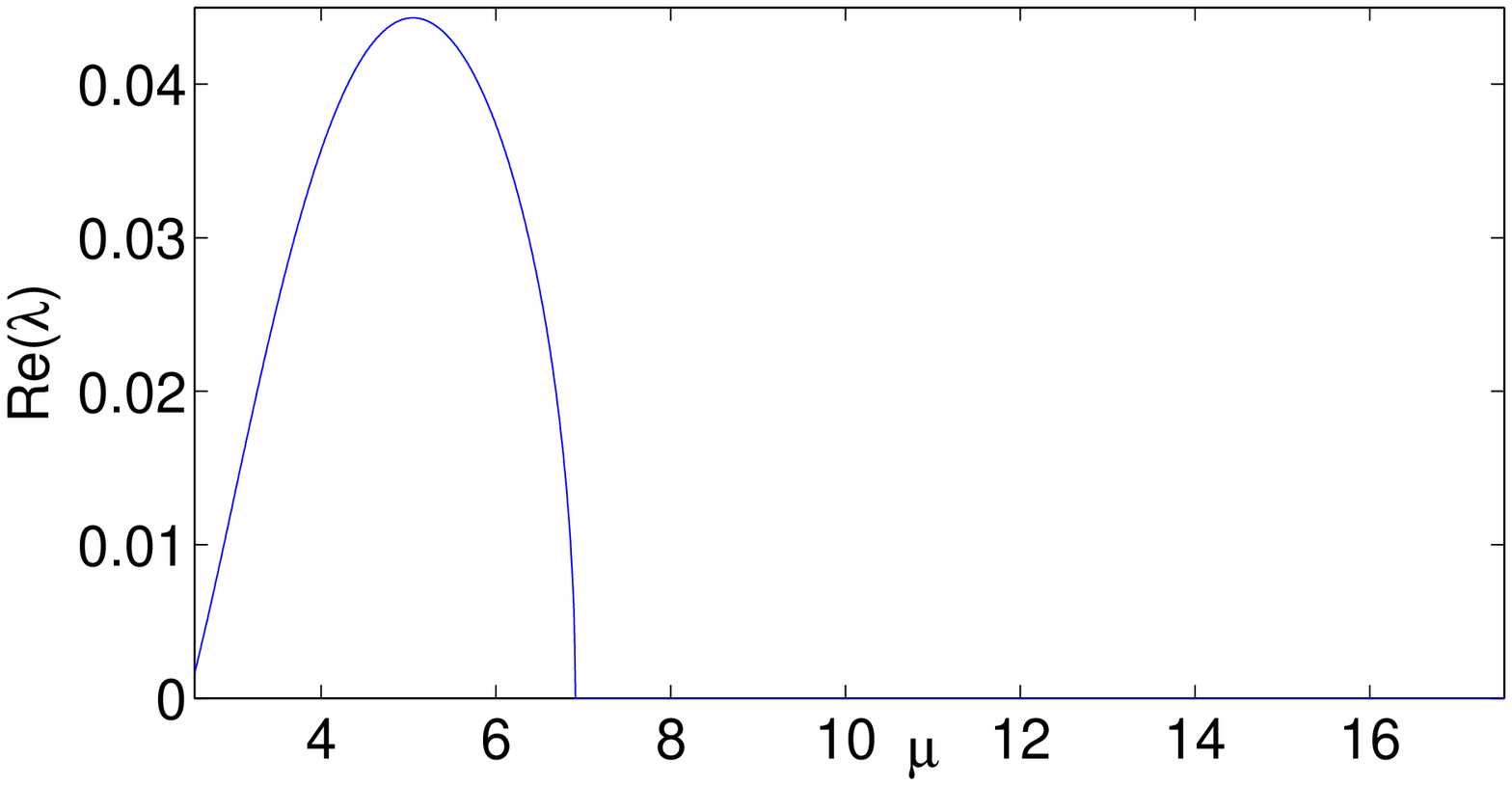}
\end{center}
\caption{Left: solid lines indicate the frequencies of linearization
around a 2-soliton solution as a function of the chemical potential
$\mu$. The dashed lines show the asymptotic limits (\ref{asym-limits})
for the frequencies around the ground state.
The dash-dotted lines indicate the asymptotic predictions for
the in-phase (lower frequency)
and out-of-phase (higher frequency) oscillations of $2$ dark solitons.
Right: real part of the unstable eigenvalue in a finite instability band
near the linear limit of $\mu = \frac{5}{2}$.} \label{dfig2}
\end{figure}

Figures \ref{dfig3} and \ref{dfig4} illustrate similar characteristics
but for the 3-soliton state. Once again the variational
prediction given by the asymptotic expansion (\ref{expansion-a-3-soliton})
provides a highly accurate estimate of the numerical inter-soliton distance
$a = a_3 - a_2 = a_2 - a_1$.

On the other hand, in this case, there exist
three frequencies associated with the relative motions of three
dark solitons, whose values can be seen to be in very good agreement with the
asymptotic expansion (\ref{asympt-fre-1}). Close to the linear limit
$\mu = \frac{7}{2}$, there exists two resonances between out-of-phase
motion of three dark solitons and the corresponding frequencies of oscillations
of the ground state. The two resonances induce instabilities of the $3$-soliton
excited states with two finite instability bands.

\begin{figure}
\begin{center}
\includegraphics[width=8cm,height=6cm]{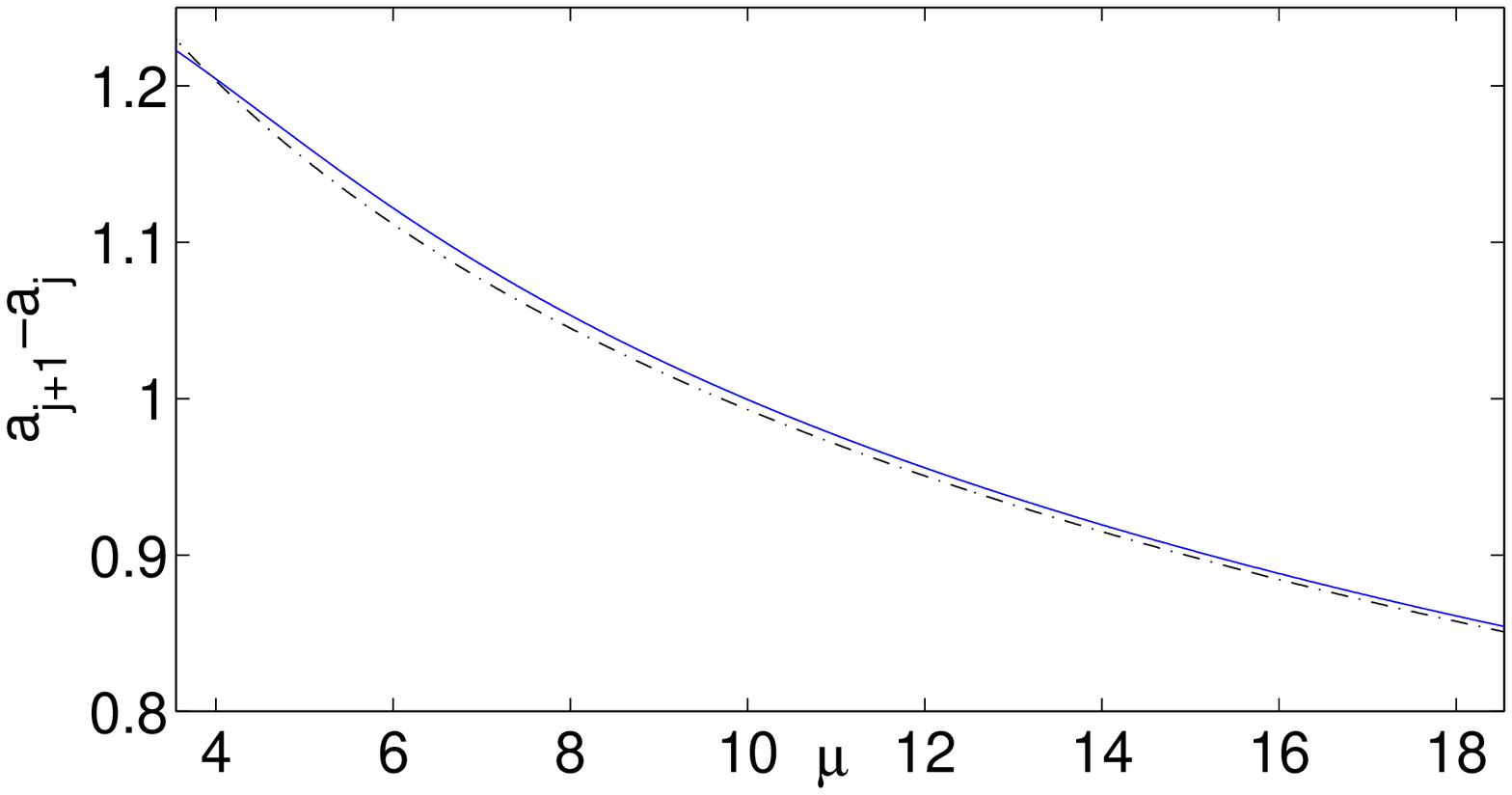}
\includegraphics[width=8cm,height=6cm]{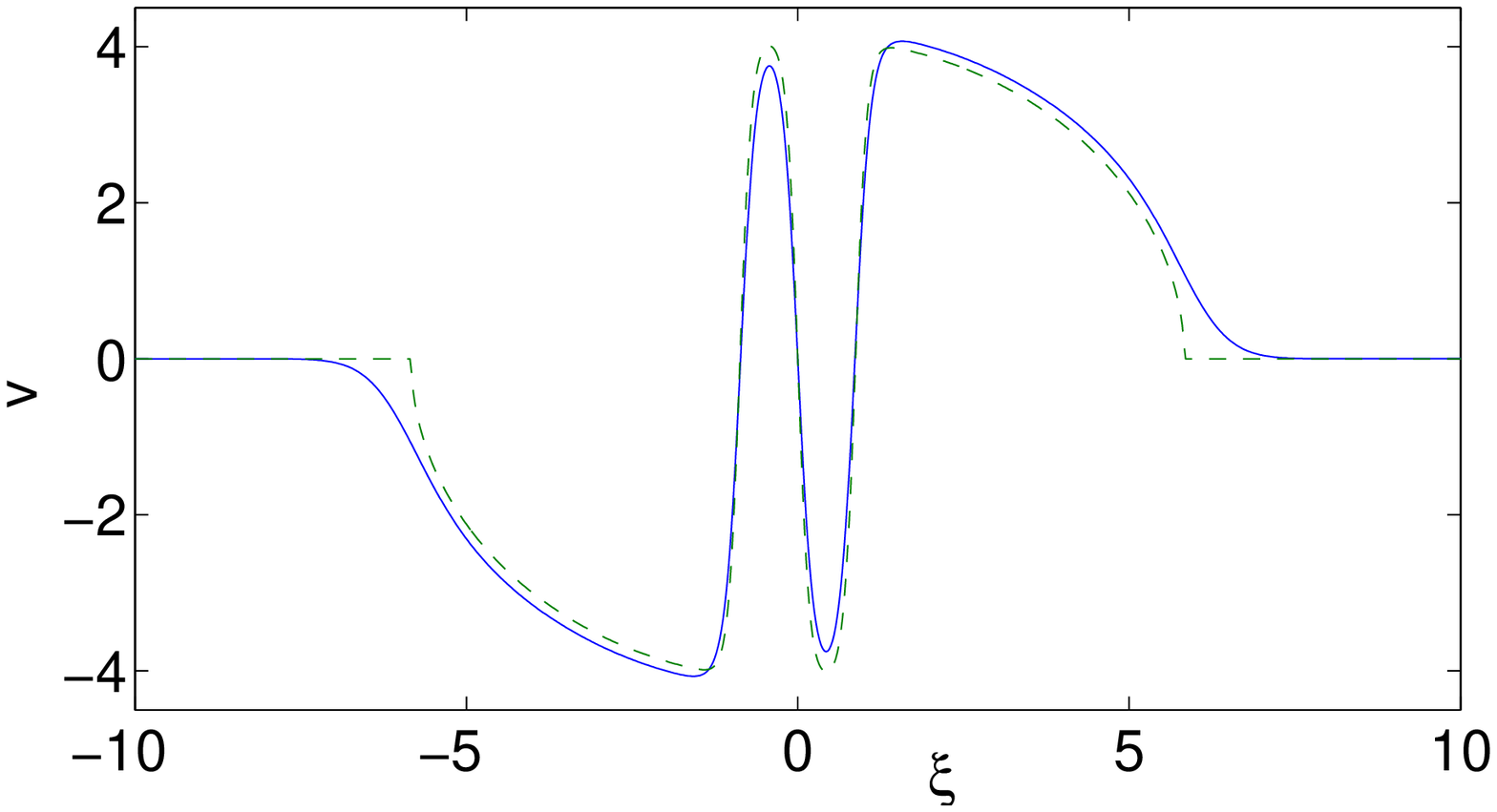}
\end{center}
\caption{Similar to Fig. \ref{dfig1} but for the 3-soliton case.
The left panel again shows the equilibrium inter-soliton distance
(solid: numerical results; dash-dotted:
asymptotic approximation), while the right shows the numerical prediction
(solid) and variational ansatz (dashed) of the
3-soliton state $v(\xi)$ for $\mu=17$.} \label{dfig3}
\end{figure}
\begin{figure}
\begin{center}
\includegraphics[width=8cm,height=6cm]{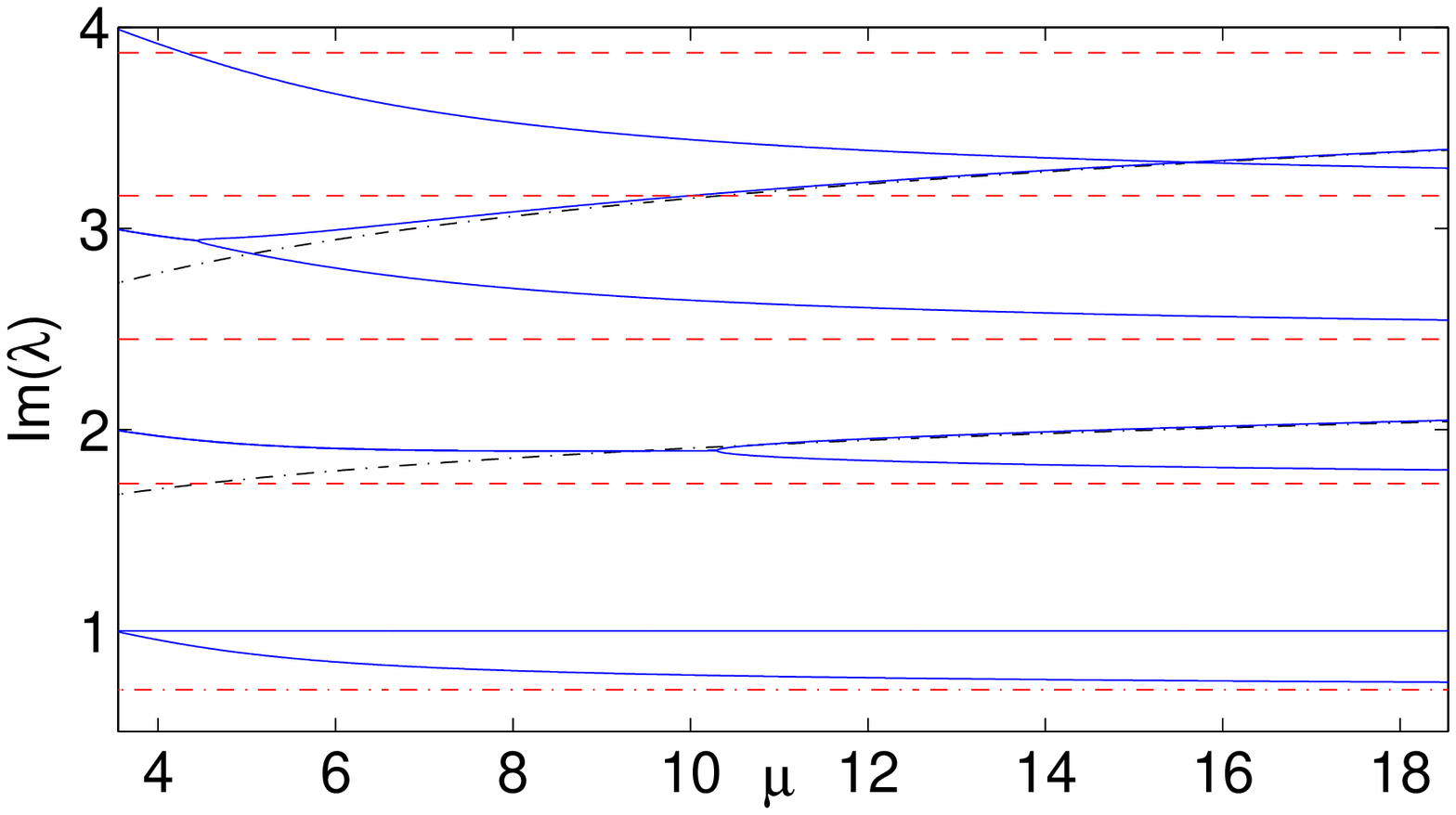}
\includegraphics[width=8cm,height=6cm]{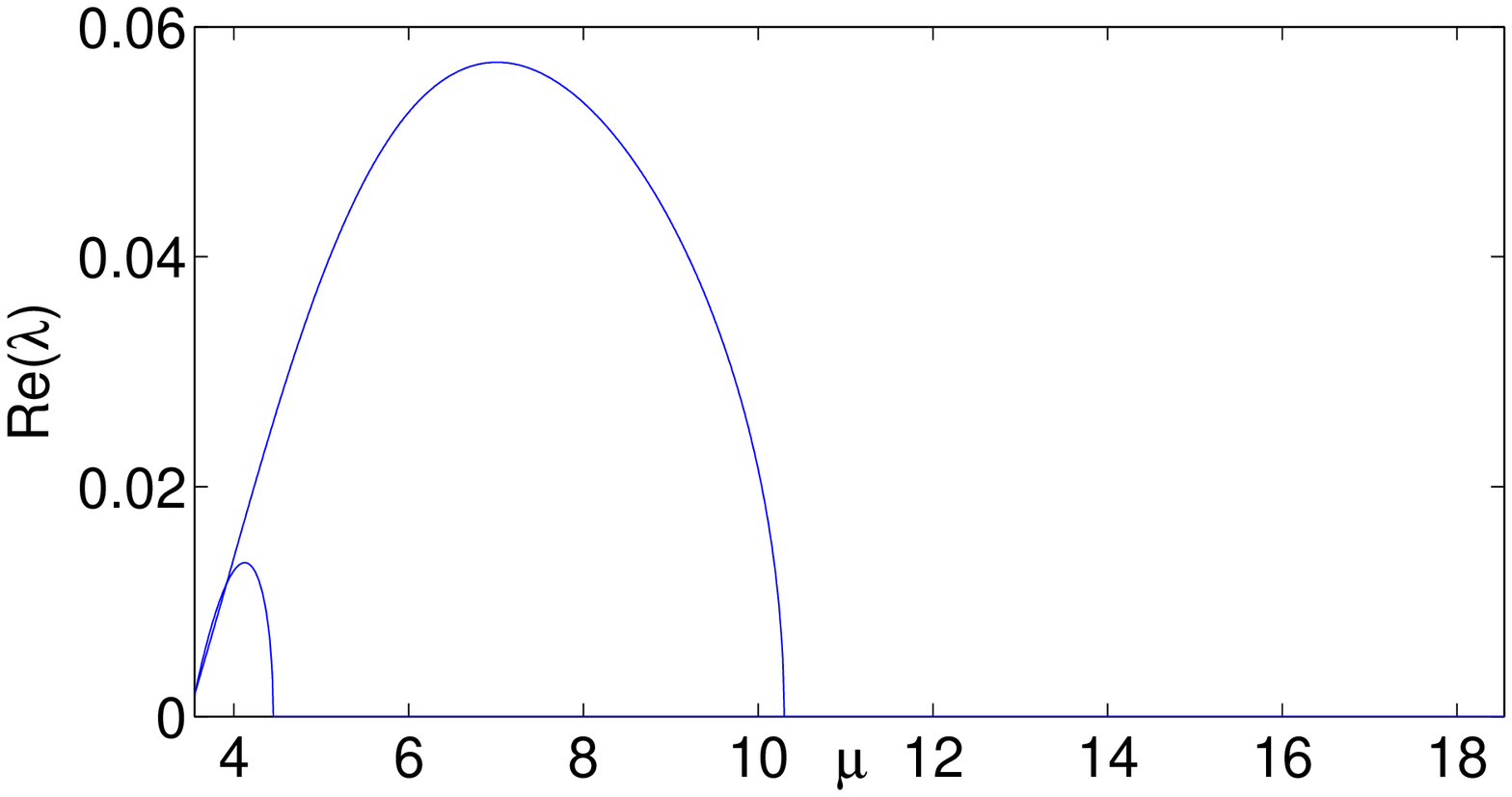}
\end{center}
\caption{Same as Fig. \ref{dfig2}, but for the 3-soliton case.
The left panel shows the numerical frequencies (imaginary parts
of the relevant eigenvalues) by solid line, the asymptotic limits for
the frequencies of the ground state by dashed line, and
the frequencies of oscillations of three dark solitons by dash-dotted line.
The right panel illustrates the real part of the unstable eigenmodes arising
close to the linear limit $\mu = \frac{7}{2}$.} \label{dfig4}
\end{figure}

The above results provide a relatively complete understanding
of the statics and dynamics of multi-soliton states within
Bose-Einstein condensates at least within the Thomas-Fermi
limit of large chemical potential. This characterization is
especially relevant presently given the recent experiments
of \cite{our2,our1} enabling the observation and robust
time-following for large timescales (of the order of hundred
milliseconds or more) of such states. However, there would
be a multitude of directions in which it would be relevant
to generalize these results, if possible. On the one hand,
extending them (analytically) to non-polynomial variants
of the Gross-Pitaevskii equation accounting for the
confinement of the condensate across tranvserse directions
would be a challenging theoretical task. Another equally
interesting direction would involve attempting to generalize
relevant notions in trying to characterize the dynamics of
vortex solitons in higher dimensional settings. These directions
are presently under consideration and corresponding results will
be reported in future publications.

\vspace{5mm}

{\bf Acknowledgments}: MC is supported by the NSERC USRA scholarship, 
DEP is partially supported by the NSERC grant, and PGK is partially supported
by NSF-DMS-0349023 (CAREER), NSF-DMS-0806762 and the Alexander-von-Humboldt
Foundation.


\begin{thebibliography}{99}

\bibitem{alfimov2} G.L. Alfimov and D.A. Zezyulin, ``Nonlinear modes for the
Gross-Pitaevskii equation--a demonstrative computation approach'', Nonlinearity {\bf 20}, 2075--2092 (2007).

\bibitem{bpa} B. P. Anderson, P. C. Haljan, C. A. Regal, D. L. Feder, L. A. Collins,
C. W. Clark, and E. A. Cornell, ``Watching dark solitons decay into vortex rings in a Bose-Einstein condensate", 
Phys. Rev. Lett. {\bf 86}, 2926--2929 (2001).

\bibitem{han2} S. Burger, K. Bongs, S. Dettmer, W. Ertmer, K. Sengstock, A. Sanpera,
G. V. Shlyapnikov, and M. Lewenstein, ``Dark solitons in Bose-Einstein condensates", 
Phys. Rev. Lett. {\bf 83}, 5198--5201 (1999).

\bibitem{revnonlin} R. Carretero-Gonz\'alez, D. J. Frantzeskakis, and
P. G. Kevrekidis, ``Nonlinear waves in Bose–Einstein condensates:
physical relevance and mathematical techniques'', Nonlinearity {\bf 21},
R139--R202 (2008)

\bibitem{GalPel1} C. Gallo and D. Pelinovsky,
``Eigenvalues of a nonlinear ground state in the
Thomas--Fermi approximation'', J. Math. Anal. Appl. {\bf 355},
495-–526 (2009)

\bibitem{GalPel2} C. Gallo and D. Pelinovsky,
``On the Thomas--Fermi ground state in a radially symmetric
parabolic trap'', preprint (2009)

\bibitem{engels} P. Engels and C. Atherton, 
``Stationary and nonstationary fluid flow of a Bose-Einstein condensate 
through a penetrable barrier", Phys. Rev. Lett. {\bf 99}, 160405-4 (2007).

\bibitem{IM} R. Ignat and V. Millot, ``The critical velocity
for vortex existence in a two-dimensional rotating Bose--Einstein
condensate'', J. Funct. Anal. {\bf 233}, 260--306 (2006)

\bibitem{IM2} R. Ignat and V. Millot, ``Energy expansion and
vortex location for a two-dimensional rotating Bose--Einstein
condensate'', Rev. Math. Phys. {\bf 18}, 119--162 (2006)

\bibitem{KP04} V. V. Konotop and L. Pitaevskii, ``Landau dynamics of a grey soliton in a trapped condensate", 
Phys. Rev. Lett. {\bf 93}, 240403-4 (2004)

\bibitem{KK} Yu.S. Kivshar and W. Krolikowski, ``Lagrangian
approach for dark solitons'', Opt. Comm. {\bf 114}, 353--362 (1995)

\bibitem{PelKev} D.E. Pelinovsky and P.G. Kevrekidis,
``Periodic oscillations of dark solitons in parabolic
potentials'', AMS Cont. Math. {\bf 473}, 159--180 (2008)

\bibitem{pethick} C.J. Pethick and H. Smith,
{\it Bose-Einstein condensation in dilute gases}, Cambridge University
Press (Cambridge, 2002).

\bibitem{PitStr} L. Pitaevskii and S. Stringari,
\textit{Bose-Einstein Condensation}, Oxford University Press (Oxford, 2003)

\bibitem{technion} I. Shomroni, E. Lahoud, S. Levy, and J. Steinhauer,
Nature Physics {\bf 5}, 193 (2009).

 \bibitem{hambcol} S. Stellmer, C. Becker, P. Soltan-Panahi, E.-M. Richter, S. D\"{o}rscher, M. Baumert,
J. Kronj\"{a}ger, K. Bongs, and K. Sengstock, ``Collisions of dark solitons in elongated Bose-Einstein 
condensates", Phys. Rev. Lett. {\bf 101}, 120406-4 (2008).

\bibitem{our2} G. Theocharis, A. Weller, J. P. Ronzheimer, C. Gross, M. K. Oberthaler, P. G. Kevrekidis, D. J. Frantzeskakis, ``Multiple atomic dark solitons in cigar-shaped Bose-Einstein condensates", 
    arXiv:0909.2122.

\bibitem{our1} A. Weller, J. P. Ronzheimer, C. Gross, J. Esteve, M. K. Oberthaler, D. J. Frantzeskakis, G. Theocharis, and P. G. Kevrekidis, ``Experimental observation of oscillating and interacting matter wave dark solitons", 
    Phys. Rev. Lett. {\bf 101}, 130401-4 (2008).

\bibitem{ZAKP} D.A. Zezyulin, G.L. Alfimov, V.V. Konotop, and
V.M. P\'erez--Garc\'ia, \textit{Stability of excited states of a Bose--Einstein condensate
in an anharmonic trap}, Phys. Rev. A {\bf 78}, 013606 (2008)

\end{thebibliography}
\end{document}